\documentclass[11pt]{article}

\usepackage[mathscr]{eucal}
\usepackage[latin1]{inputenc}
\usepackage[T1]{fontenc} 
\usepackage[english]{babel}
\usepackage{amscd,amssymb,amsfonts}
\usepackage{latexsym}
\usepackage{makeidx} 
\usepackage{ifthe
n}
\usepackage{graphicx}
\usepackage{amsmath}
\usepackage{fancyhdr}
\usepackage{multirow}
\usepackage{eurosym}
\usepackage{textcomp} 
\usepackage{chngpage}
\usepackage{lscape}
\usepackage{layout}
\usepackage{indentfirst}
\usepackage{theorem}
\usepackage{url}

\usepackage[vcentering,pdftex,dvips]{geometry}

\numberwithin{equation}{section} 
\newtheorem{theorem}{Theorem}

\newtheorem{corollary}[theorem]{Corollary}

\newtheorem{definition}[theorem]{Definition}

\newtheorem{lemma}[theorem]{Lemma}

\newtheorem{proposition}[theorem]{Proposition}
\newtheorem{remark}[theorem]{Remark}

\newenvironment{proof}[1][Proof]{\noindent\textbf{#1.} }{\ \rule{0.5em}{0.5em}}

\def\f{{\cal F}}

\title{The arbitrage-free Multivariate Mixture Dynamics Model:\\ Consistent single-assets and index volatility smiles}
\author{
Damiano Brigo\thanks{Dept. of Mathematics, Imperial College, London. damiano.brigo@imperial.ac.uk} \hspace{.5cm} Francesco
Rapisarda\thanks{Method Investments \& Advisory Ltd. This paper reflects solely the Author's personal opinion and does not represent the opinions of the author's  employers, present and past, in any way. francesco\_rapisarda@ymail.com} \hspace{.5cm} Abir Sridi\thanks{Dept. of Mathematics, Universit\'{e} Paris I Panth\'{e}on-Sorbonne, Paris. abir.sridi@malix.univ-paris1.fr}}

\date{\small{First version: 1 Feb 2012. This version: 23 Sept 2014. First posted on SSRN \& arXiv on Feb 2013}}

\begin{document}
\maketitle

\vspace{-1cm}

\begin{abstract}

We introduce a multivariate diffusion model that is able to price
derivative securities featuring multiple underlying assets.
Each asset volatility smile is modeled according to a density-mixture dynamical model while the
same property holds for the multivariate process of all assets, whose density is a mixture of multivariate basic densities.
This allows to reconcile single name and index/basket volatility smiles in a consistent framework.
Our approach could be dubbed a multidimensional local volatility approach with vector-state dependent diffusion matrix.
The model is quite tractable, leading to a complete market and not requiring Fourier techniques for calibration and dependence measures, contrary to multivariate stochastic volatility models such as Wishart.
We prove existence and uniqueness of solutions for the model stochastic differential equations,
provide formulas for a number of basket options, and analyze the dependence structure of the model in detail by deriving a number of results on covariances, its copula function and rank correlation measures and volatilities-assets correlations. 
A comparison with sampling simply-correlated suitably
discretized one-dimensional mixture dynamical paths is made, both in terms of option pricing and of dependence, and first order expansion relationships between the two models' local covariances are derived. 
We also show existence of a multivariate uncertain volatility model of which our multivariate local volatilities model is a Markovian
projection, highlighting that the projected model is smoother and avoids a number of drawbacks of the uncertain volatility version. We also  show a consistency result where the Markovian projection of a geometric basket in the multivariate model is a univariate mixture dynamics model. A few numerical examples on basket and spread options pricing conclude the paper. \bigskip


\bigskip

{\bf AMS classification codes}: 60H10, 60J60, 62H20, 91B28, 91B70


\end{abstract}

\tableofcontents

\newpage

\section{Introduction}

It has been known for a long time that the Black--Scholes geometric
Brownian motion model \cite{black_scholes} does not price
all European options quoted on a given market in a consistent way.
In fact, this model lies on the fundamental assumption that the
asset price volatility is a constant. In reality, the implied
volatility, namely the volatility parameter that, when plugged
into the Black--Scholes formula, allows to reproduce the market
price of an option, generally shows a dependence on both the option
maturity and strike. If there were no dependence on strike one could
extend the model in a straightforward fashion by allowing a
deterministic dependence of the underlying's instantaneous
volatility on time, so that the dynamics could be represented by the
following stochastic differential equation (SDE):
\begin{equation}
dS_t = \mu S_t dt + \sigma_t S_t dW_t,
\end{equation}
$\sigma_t$ being the deterministic instantaneous volatility referred
to above. In that case, reconstruction of the time dependence of
$\sigma_t$ would follow by considering that, if $v(T_i)$ denotes the
implied volatility for options maturing at time $T_i$, then
\begin{equation}
v(T_i)^2 T_i = \int_0^{T_i} \sigma_s^2 ds.
\end{equation}

Implied volatility however does indeed show a strike dependence; in
the common jargon, this behavior is described with the term {\em
smile} whenever volatility has a minimum around the forward asset price
level, or {\em skew} when low--strike implied volatilities are
higher than high--strike ones. In the following we will loosely
speak of both effects as ''volatility smile''.

In recent years, many researches have tried to incorporate the smile
effect into a consistent theory. Several streams of investigation can be identified in a univariate setting. We do not aim at completeness in the following review, but just present a few  relevant examples.

A first approach is based on assuming an {\em alternative explicit
dynamics} for the asset--price process that by construction ensures
the existence of volatility smiles or skews. Typically, in this dynamics the diffusion coefficient
of the asset price is a deterministic function of the asset price itself and of time.
This is referred to with the term ``local volatility". Examples include the CEV
process proposed by Cox \cite{cox} and Cox and Ross \cite{cox_ross}.
A different example is the displaced diffusion model by Rubinstein \cite{rubi83}.
In general the alternative explicit dynamics does not reproduce accurately
enough the market volatility structures, since it is based on quite stylized dynamics, with the
mixture dynamics exception we will see in a moment.

A second approach is based on the assumption of a {\em continuum of
traded strikes} \cite{Breeden}. This was extended yielding an
explicit expression for the Black--Scholes implied volatility as a
function of strike and maturity
\cite{Derman,DermanKani,Dupire,Dupire1}. This approach however needs
a smooth interpolation of option prices between consecutive traded
strikes and maturities. Explicit expressions for the risk--neutral
stock price dynamics were also derived by minimizing the relative
entropy to a prior distribution \cite{Avellaneda} and by assuming an
analytical function describing the volatility surface
\cite{BrownRandall}.

Another approach is {\em an incomplete market} approach, and
includes stochastic volatility models
\cite{Heston,HullWhite,Tompkins}, jump--diffusion models
\cite{Prigent} and more recently stochastic-local volatility models \cite{labordere} combining
local and stochastic volatility.

A further approach consists of finding the risk--neutral distribution
on a {\em lattice} model for the underlying that leads to a best fit
of the market option prices subject to a smoothness criterion
\cite{BrittenJones,Jackwerth}. This approach has the drawback of
being entirely numerical.

A number of the above approaches is described for the foreign exchange market in Lipton \cite{liptonbook},
see also Gatheral  \cite{gatheral} who deals further with volatility surfaces parameterization.
Recent literature also focused on both short-- and long--time asymptotics for volatility models: we just cite \cite{GaEtAl} as a reference for small time asymptotics in local volatility models, and \cite{FJ09} for large maturities asymptotics in the well known Heston stochastic volatility model, while pointing out that the volatility asymptotics literature is much broader.

In general the problem of finding a risk--neutral distribution that
consistently prices all quoted options is largely undetermined. A
possible solution is given by assuming a particular parametric
risk--neutral distribution dependent on several, possibly
time--dependent, parameters and use the latter in conjunction with a
calibration procedure to the market option prices. In a number of papers, Brigo, Mercurio, Rapisarda and Sartorelli
\cite{general_mixture_diffusion,mixture1,mixture2,mixtureFX,sartorelli} proposed a family of models that carry on
dynamics leading to a parametric risk--neutral distribution flexible enough for practical purposes.
It is relatively straightforward to postulate a mixture distribution at a given point in time, but it is less so finding a stochastic process that is consistent with such distribution and whose stochastic differential equation has a unique strong solution. This is the approach adopted by the above papers. This family of models is summarized for example in Musiela and Rutkowski \cite{musiela}, or Fengler \cite{fengler}, see also Gatheral \cite{gatheral}.
Formally, this is part of the alternative explicit dynamics
branch of models but is typically much richer than the models listed above, leading to a practically exact fit of the volatility smile while retaining analytical tractability.

\bigskip

The aim of this paper is to incorporate the effect of the
volatility smile observed on the market when pricing and hedging
{\em multiasset} securities, while retaining sensible single--asset volatility structures.
A whole lot of such structured
securities is nowadays offered to institutional and retail
investors, in the form of options on baskets of stocks/FX rates and
on combinations of forward interest rates such as e.g.
European/Bermudan swaptions. In our approach we remain within a lognormal-mixture 
local volatility model for the individual assets composing the
underlying of the option (be it a basket of stocks or a swap rate)
that has proved to be quite effective in accounting for the observed
single--assets' smiles, but we move one step beyond the na\"ive ``Brownian correlation" 
way to connect these univariate models when writing the joint multi-asset dynamics.
Indeed, given univariate local volatility (one dimensional diffusion-) models for each asset, a basic approach
is introducing instantaneous correlations across the Brownian shocks of each asset, leading to what we call the Simply
Correlated Mixture Dynamics (SCMD).
For practical implementation, one would then discretise the one-dimensional single--asset SDEs through, say, Euler or higher
order numerical schemes \cite{kloeden_platen}, feeding correlated instantaneous Brownian shocks into the scheme.
In this paper we adopt a different approach and we incorporate statistical dependence 
in a new scheme that enjoys analytic multivariate densities and a
fully analytic multivariate dynamics through a state dependent non-diagonal
diffusion matrix. In so doing we are able to sample a new manifold
of instantaneous covariance structures (and a new manifold of
dynamics) which ensures full compatibility with the individual
volatility smiles and overcomes the difficult problems created by the
lack of closed form formulas for prices and sensitivities on multi-asset securities. We call
the resulting model Multi Variate Mixture Dynamics (MVMD) and prove existence and uniqueness of the solution for its multivariate stochastic differential equation.

The traditional approach for pricing European--style derivatives on
a basket of the  multidimensional underlying, in a SCMD type model, uses a Monte Carlo method that can
be very slow as it involves intensive time discretization, given that correlation can only be introduced at local shocks level.
With this paper we fill this substantial gap in option pricing and provide, with MVMD type models, a semi-analytic solution to the
option pricing problem where the price can be quickly and accurately evaluated, something that practitioners value greatly, especially in the Risk Management analytics area.
%
The level of tractability in MVMD for both single assets and indices/baskets is much higher than with multivariate stochastic volatility models such as Wishart models,
for which we refer for example to \cite{gurieroux07,fonseca07} and references therein. This tractability extends to a lot of dependence measure calculations, as we shall see shortly, which are fundamental in a multi-asset model. Furthermore,
the MVMD model leads to a complete market and hedging is much simpler. It is practically a tractable and flexible multivariate local volatility model that has the potential to consistently calibrate univariate and index volatility smiles through a rich but at the same time transparent parameterization of the dynamics.

In multi-asset models the transparency on statistical dependence structures and their dynamics is fundamental. 
This is why we study and calculate in closed form instantaneous correlations between assets, terminal correlations, average correlations, rank correlations, squared volatility - assets correlations, and the whole copula function of the MVMD model. Such explicit study and formulas are not available in SCMD or Wishart models.  We also derive an expansion of the local covariance in MVMD, showing that the first term in the expansion coincides with the analogous term in SCMD. As a form of comparison between MVMD and SCMD, we look at Kendall's tau rank correlation measures across assets in detail, as implied by the two different models when the same parameters are chosen. 


We then introduce a Multivariate Uncertain Volatility
Model (MUVM). We show that the MVMD model is a Markovian projection of the MUVM.
MUVM thus gives the same European option prices as MVMD and can be used instead of MVMD
 to price European options also in the multivariate setting.
MUVM features the same dependence structure as the MVMD model. The
related copula is a mixture of multivariate copulas that are each a
standardized multivariate normal distribution with an appropriate
correlation matrix and marginals.
Despite these similarities, the MUVM model is less smooth and convincing than the MVMD model.
The fact that the uncertaintly of volatiltiy needs to be realized instantly in a very near future is unrealistic and may lead to
problems when hedging with the model and when dealing with early exercise products, especially when exercise is considered near the date of realization of the uncertain volatility. Hence while we show the Markovian projection property as an interesting mathematical result, we recommend usage of MVMD rather than MUVM for products where the two models produce different prices.


We further point out a result on correlation between assets and their instantaneous variances (squared volatilities) and covariances. A drawback of local volatility models is that they cannot decorrelate assets and volatilities, since the latter are deterministic functions of the assets themselves. However, as pointed out in \cite{general_mixture_diffusion} for the univariate case, in the MVMD model we have complete decorrelation between assets and instantaneous covariances. While this is surprising at first sight, given that all instantaneous covariances are deterministic functions of the joint assets, it becomes more intuitive when thinking about the relationship with MUVM, and is the best approximation MVMD can attain for its non-Markovian originator MUVM, where instantaneous covariations and assets Brownian shocks are fully independent. 


We further highlight a Markovian projection property for the basket dynamics implied by MVMD. We consider the Markovian projection of the  Geometric average basket dynamics implied by MVMD on one dimensional diffusions. We find that the multivariate mixture dynamics for the basket components induces a univariate lognormal mixture dynamics for the basket, in a consistency result that can be used to price European basket options on the geometric basket in fully closed form via a Black Scholes formula. As far as the geometric average can be considered as a good proxy for the arithmetic one \cite{Vorst}, the method could be used for standard basket options, or at the very least serve as a control variate result for the one-shot simulation needed to price an option on an arithmetic basket. In the context of geometric baskets, no other similar consistency results are known for multivariate models. 


We then introduce option pricing for basket options and spread options, deriving semi-closed form formulas or one-shot simulation schemes for MVMD against multi-step Monte Carlo simulation for SCMD with analogous parameters.   
In the final part of this work, in order to develop a feel for the performance of our approach, we test it on a few cases, including arithmetic and geometric averages (weighted) baskets and spread options. We compare the prices generated by MVMD to those obtained by the SCMD model with analogous parameters, and conclude that options prices may not reflect the difference in dependence structures between the two models even for payoffs, such as spread options, that should depend heavily on the model dependence structure. 

\bigskip

The paper is organised as follows. In Section \ref{sect:MD_model},
we present a brief review of the approach to single--asset smile
modeling that has been developed in
\cite{general_mixture_diffusion,mixture1,mixture2,mixtureFX}. In Section
\ref{sect:formulation_problem}, we provide examples of typical
securities that need a multivariate setting for proper pricing.
Section \ref{sect:extension_MD_to_multivariate_problems} considers
the extension of the single--asset model to the multivariate
framework with a thorough discussion of the implications for the
dynamics stemming from a na\"{\i}ve approach (SCMD) and from ours
(MVMD). 
In Section \ref{sec:dependencest} we provide a number of results on the dependence structure in the MVMD and SCMD models. 
In Section \ref{sect:Markovian_projection}, we introduce a new model that we call
"Multivariate Uncertain Volatility Model" so that our model is a
multivariate Markovian projection of it. We also show a consistency result for the Markovian projection of the geometric basket dynamics in the MVMD model, that turns out to be a univariate mixture dynamics model. 
In Section \ref{sec:optpric} we explain how to price arithmetic, geometric and spread basket options in MVMD and how this is much easier than with SCMD, deriving the relevant formulas. 
In Section \ref{sect:numerical_results_pricing}, we illustrate
the results of pricing European option on a weighted arithmetic average of the underlying assets with positive weights,
European spread option and European option on weighted geometric
average in both MVMD and SCMD frameworks and we compare the results.
Conclusions and suggestions for future research are given in the
final section.

%
%
%
%
%
%

\section{The Mixture Dynamics (MD) Model}\label{sect:MD_model}

For a maturity $T>0$ denote by $P(0,T)$ the price at time $0$ of
the zero-coupon bond maturing at $T$. Let $(\Omega,\f,\mathbb{P})$ be a probability space with a filtration $({\f}_t)_{t\in[0,T]}$ that is $\mathbb{P}$-complete and satisfying the usual conditions. We assume the existence of a measure $\mathbb{Q}$ equivalent to $\mathbb{P}$ called the risk--neutral or pricing measure, ensuring arbitrage freedom in the classical setup, for example, of Harrison,  Kreps and Pliska \cite{harrison&kreps,harrison&pliska}.
At times, it will be convenient to use the $T$--forward risk-adjusted measure $\mathbb{Q}^T$ rather than $\mathbb{Q}$.

The MD model is based on the hypothesis that the dynamics of the asset underlying a given option
market takes the form
\begin{equation}
dS(t) = \mu(t) S(t) dt + \nu(t,S(t)) S(t) dW(t) \label{localVol}
\end{equation}
under $\mathbb{Q}$ with initial value
$S_0.$ Here, $\mu$ is a deterministic time function, $W$ is a standard $\mathbb{Q}$
Brownian motion and $\nu$ (the "local volatility") is a well behaved
deterministic function. In order to guarantee the existence of a
unique strong solution to the above SDE, $\nu$ is assumed to be locally Lipschitz, uniformly in $t$, and to satisfy
the linear growth condition
\begin{equation}
\nu^2(t,x) x^2 \le L (1+x^2) \hspace{.2cm} \mbox{uniformly in $t$}
\end{equation}
for a suitable positive constant $L$.

Consider $N$ purely instrumental diffusion processes $Y^i(t)$ with dynamics
\begin{equation}\label{lognormal_EDS}
d Y^i(t) = \mu (t) Y^i(t) dt + v^i(t, Y^i(t))Y^i(t) dW(t)
\end{equation}
with initial value $Y^i(0),$ marginal densities $p_t^{i}$ and with $v_i$
satisfying locally Lipschitz and linear growth conditions, where each $Y^i(0)$ is set to
$S(0)$.

\begin{remark}
The reader should not interpret the $Y^i$ as real assets. They are just instrumental processes that will be used to define mixtures of densities with desirable properties.
\end{remark}
The marginal density $p_t$ of $S(t)$ is assumed to be
representable as the superposition of the instrumental processes densities $p_t^{i}$
\cite{mixture1,mixture2,mixtureFX}:
\begin{equation}
p_t = \sum_i \lambda^i p_t^{i} \hspace{.2cm} \mbox{with}
\hspace{.2cm} \lambda^i \ge 0, \forall i \hspace{.2cm} \mbox{and}
\hspace{.2cm} \sum_i \lambda^i = 1. \label{mixtureOne}
\end{equation}

The problem of characterizing $\nu$ can then be cast in the
following form: is there a local volatility $\nu$ for Eq.
(\ref{localVol}) such that Eq. (\ref{mixtureOne}) holds? Purely
formal manipulation of the related Kolmogorov forward equation
\begin{equation}
\frac{\partial p_t}{\partial t} + \frac{\partial}{\partial x} (\mu x
p_t) - \frac{1}{2} \frac{\partial^2}{\partial x^2} (\nu^2(t,x) x^2
p_t) = 0
\end{equation}
and of analogous equations for the $p_t^{i}$'s shows that a
candidate $\nu$ is
\begin{equation}\label{nu}
\nu(t,x) = \sqrt{ \frac{\sum_{i=1}^N \lambda^i v^i(t,x)^2
p_t^{i}(x)} {\sum_{i=1}^N \lambda^i p_t^{i}(x)} }.
\end{equation}
We may now introduce the following
\begin{definition} {\bf General MD model.} The general single-asset Mixture Dynamics (MD) candidate model is the model given by equations \eqref{localVol} and \eqref{nu}. If the model equation admits a unique solution and if the related Kolmogorov forward equation admits a unique solution, then the density of the model is a mixture according to Equation
\eqref{mixtureOne}, where the $p^i$ terms are the densities of the instrumental diffusion processes \eqref{lognormal_EDS}.
\end{definition}

An important consequence of the above construction is the following

\begin{proposition} Assume that the model (\ref{localVol},\ref{nu}), with $p^i_t$ from \eqref{lognormal_EDS}, admits a unique strong solution and that the related Kolmogorov forward equation admits a unique solution. Then the
pricing of European options on $S$ is simply a linear-convex combination with weights $\lambda^i$ of the option prices under the instrumental asset dynamics (\ref{lognormal_EDS}). Similarly for the Greeks at time 0.
\end{proposition}

In other terms, let $O$ be the
value at $t = 0$ of an European option with strike $K$ and maturity
$T$. $O$ is given by $O = \sum_{i=1}^N \lambda^i O_i$; where $O_i$
is the European price associated to the hypothetical instrumental dynamics (\ref{lognormal_EDS}). The
option price $O$ can be viewed as the weighted average of the European
option prices written on the processes $Y^i$. Due to linearity of differentiation,
the same convex combination applies to
all option Greeks. As a consequence, if the basic densities
$p_t^{i}$ are chosen so that the prices $O_i$ are computed
analytically, one finds an analytically tractable model.

\bigskip

The most natural choice for the $(Y^i, v_i, p_t^{i})$ triplet is :
\begin{equation}\label{sgherlo}
\left\{
\begin{array}{l}
Y^i(0)=S(0)\\
\\
v_i(t,x) = \sigma^i(t) \\
\\
V^i(t) = \sqrt{\int_0^t \sigma^i(s)^2 ds} \\
\\
p_t^{i}(x) = \frac{1}{\sqrt{2 \pi} x V^i(t)} \exp\left[ -\frac{1}{2
V_i^2(t)} \left( \ln\left(\frac{x}{S(0)}\right) - \mu t +
\frac{1}{2} V^i(t)^2 \right)^2 \right] =: \ell^i_t(x)
\end{array}
\right.
\end{equation}
with $\sigma^i$ deterministic ({\it  lognormal mixture dynamics, LMD}).

Brigo and Mercurio \cite{mixture2} proved that, with the above choice 
and additional nonstringent assumptions on the $\sigma_i$, the
corresponding dynamics for $S_t$ indeed admits a unique strong
solution. A greater flexibility can also be achieved by shifting the
auxiliary processes' density by a carefully chosen deterministic
function of time (still preserving risk--neutrality). This is the
so--called {\em shifted lognormal mixture dynamics} model \cite{mixtureFX}.

\begin{theorem}\label{th:LMDE} {\bf Existence and uniqueness of solutions for the LMD model}. Assume that all the real functions $\sigma^i(t)$, defined on the real numbers $t \ge 0$, are once continuously differentiable and bounded from above and below by two positive real constants. Assume also that in a small initial time interval $t \in [0, \epsilon]$, $\epsilon >0$, the functions $\sigma^i(t)$ have an identical constant value $\sigma_0$. Then the Lognormal Mixture Dynamics model (LMD) defined by Equations (\ref{localVol},\ref{sgherlo}), namely
\[  \]
\begin{equation}\label{eq:dcmix}  d S_t = \mu(t) S_t dt + s(t,S_t) S_t dW_t, \ \ S_0, \ \  s(t,x)  =  \left(\frac{\sum_{k=1}^N
\lambda^{k} \sigma^k(t)^2
\ell^k_t(x)}{\sum_{k=1}^N \lambda^{k}
\ell^k_t(x)}\right)^{1/2},
\end{equation}
admits a unique strong solution and the Kolmogorov equation for its density admits a unique solution satisfying \eqref{mixtureOne}, which is in this case a mixture of lognormal densities, leading to option prices that are linear combinations of Black-Scholes prices.
\end{theorem}

In \cite{mixture1,mixture2} it is pointed out that the squared diffusion coefficient $s(t,x)^2$ defined in  \eqref{eq:dcmix} can be considered as a state dependent weighted (convex combination) average of the basic squared volatilities $(\sigma^k)^2$ and that if the latter are uniformly bounded so is $s$.

The above description gives a sufficient basis for presenting our
generalization of the LMD to the
multivariate setting, as before at first on the basis of pure
formal manipulations, and then with full rigor, with the specific aim
of finding a method to infer the ``implied volatility'' of a basket
of securities from  the individual components and/or an explicit
dynamics for the multi-asset system. Later in the paper, formal
proofs of the general consistency of the model and of the existence
and uniqueness of the solution to the multivariate version of Eq.
(\ref{localVol}) will be provided.

\section{Options on Baskets: Motivating multivariate models}\label{sect:formulation_problem}
A generalization of LMD to the
multivariate setting aims to be able to compute the smile effect on
the implied volatilities for exotic options depending on more than
one asset, such as a basket options. Clearly, analogous techniques apply to indices.
\subsection{Basket option}
A basket option is an option whose payoff is linked to a portfolio
or "basket" of underlying assets. We can distinguish two types of
basket option:
\begin{itemize}
\item An option of weighted arithmetic average of the
underlyings:
\begin{equation}
B_t = \sum_{k=1}^n w_k S_k{(t)}, \label{basket}
\end{equation}
where $B$ is called an ``arithmetic basket";
\item An option of weighted geometric average of the underlyings:
\begin{equation}\label{weighted_geometric_average}
B_t = \left[\prod_{k = 1}^n S_k(t)^{w_k}\right]^{\frac{1}{w_1
+\ldots+ w_n}}
\end{equation}
where $B$ is called a geometric basket. 
\end{itemize} where
$S_k$ is the $k$--th component of the basket. Typically the basket
is consisting of several stocks, indices or currencies. Less
frequently, interest rates are also possible ($S_k$ could represent
a forward rate process $F_k$ in the Libor Market Model (LMM) and
instead of (\ref{basket}) we could have a more complicated
expression representing a swap rate).

Such options have the most varied nature: from the plain European
call/put options on the value of the basket at maturity $T$, to
options somewhat more complicated, such as Asian options on the
basket, Himalaya options, rainbow options and so on.

The weights $(w_k)_k$ in (\ref{basket}) can be negative. When the
basket (\ref{basket}) contains short positions it is called spread
and the option known as a spread option is written on the difference
of underlying assets. The weights $(w_k)_{k = 1,\ldots,n}$ in
(\ref{weighted_geometric_average}) are positive.

It is instructive to view a basket option as a standard derivative
on the underlying instrument whose value at time $t$ is the basket
$B_t$ so defined. 

\subsection{European options pricing}
Let us assume that interest rates are constant and equal to $r > 0$
. We also assume the existence and uniqueness of a risk--neutral
pricing measure $\mathbb{Q}$ that is equivalent to $\mathbb{P}$ under which discounted asset prices are
martingales, implying the absence of arbitrage ($\mathbb{Q}$ is also equal to $\mathbb{Q}^T$ as interest rates are assumed to be deterministic).

\bigskip

According to the Black--Scholes pricing paradigm
\cite{harrison&kreps,harrison&pliska}, the price $\Pi$ of an
European option at initial time $t = 0$ is given by the risk-neutral
expectation:
\begin{equation}\label{European_option_price}
\Pi = e^{- r T} \mathbb{E}\left \{\ [\omega\ (B_T-K)]^+\right\}
\end{equation}
where the exponential factor takes care of the discounting and
$\omega = ± 1$ for a call/put respectively. $B_T$ is the underlying
instrument (can represent the value of the basket) at maturity $T$,
$K$ is the strike.

\bigskip

The fundamental difficulty in pricing basket options on a weighted arithmetic average of a basket is to determine
the distribution of the sum of underlying asset prices.
Let us consider the basket of securities of Eq. (\ref{basket}). Several approximation methods have been proposed for options on it when each $S_k$ follows a geometric Brownian motion. Usually the basket value (\ref{basket}) is approximated by the lognormal distribution. Recall
that here we consider baskets with possibly negative weights, such
as spreads. Hence, we cannot approximate the distribution of $B_t$
by a lognormal distribution, since such a basket can have negative
values or negative skewness. However, Brigo and Masetti \cite{brigo_masetti} in a LIBOR market model setting and later Borovkova, Permana and Weide \cite{borovkova_permana_weide} show that a more general three-parameter
family of lognormal distributions: shifted, negative and negative
shifted lognormal, can be used to approximate the distribution of a
general basket. The shifted lognormal distribution is obtained by
shifting the regular lognormal density by a fixed amount along the
$x$-axis, and the negative lognormal - by reflecting the lognormal
density across the $y$-axis. The negative shifted lognormal
distribution is the combination of the negative and the shifted one.
Note that this family of distributions is flexible enough to
incorporate negative values and negative skewness: something that
the regular lognormal distribution is unable to do. However, by using these
approximations we do not take into account the internal composition
of the basket value in terms of underlying assets having each its own dynamics.
This approach structurally cannot take into account any smile effect
on the individual underlyings' volatility, and therefore on the "basket
volatility".

\bigskip

In the following we will tackle the problem in a rigorous way,
through the generalization of the dynamical model of Eqs. (\ref{localVol},\ref{sgherlo}) that has proven to
perform quite well on some markets
\cite{mixture1,mixture2,mixtureFX} and that is under extension to
the equity markets case.

\section{Multivariate extensions of the MD model}\label{sect:extension_MD_to_multivariate_problems}

To fix ideas, suppose we are faced with the following problem: we
want to price an option maturing at $T$ on the basket of $n$ securities
given by Eq. (\ref{basket}) or Eq. (\ref{weighted_geometric_average}).
Each of these $n$ securities will have a ''smiley'' volatility
structure, and we expect the basket to show a smile
in its implied volatility, too.

\bigskip

Through Eqs. (\ref{localVol}--\ref{sgherlo}) we now have a piece of
machinery that allows us to calibrate an LMD to each  implied
volatility smile structure of the individual component $S_k$ of the
basket. Suppose we have already calibrated the individual LMDs to
such smile surfaces, thus finding the LMD local volatilities governing
the dynamics of each $S_k$.
We denote by $Y_k^1,\ldots,Y_k^N$ the instrumental processes for asset $S_k$.
Namely, for each asset  in the basket we have a family of instrumental processes
like (\ref{lognormal_EDS}) that refer to that specific asset mixture distribution, each (\ref{lognormal_EDS}) being specialized according to
Equation \eqref{sgherlo}.
To guide the reader through notation, we recall as a simple convention that for us upper indices in general denote a component in the mixture,
whereas lower indices denote different assets. So for example $\sigma_k^h$ will refer to asset $S_k$ and to the $h$-th component of the mixture, whereas the density of $Y_i^k$ at time $t$ will be denoted by $\ell_{i,t}^k$.

We are now interested in connecting these univariate LMD models $S_1,\ldots,S_n$ into a multivariate model
that embeds statistical dependence among the different asset. The most immediate
way to do this is to introduce a non-zero quadratic covariation between the Brownian motions driving the LMD models for $S_i$ and $S_j$ respectively.

\subsection{Simply Correlated Mixture Dynamics model}

\begin{definition} {\bf SCMD Model}.
We define the Simply Correlated multivariate Mixture Dynamics (SCMD) model
for $\underline{S}= [S_1,\ldots,S_n]$  as a vector of univariate LMD models, each satisfying Theorem
\ref{th:LMDE} with diffusion coefficients $s_1,\ldots,s_n$ given by Formula  \eqref{eq:dcmix} and densities $\ell_1,\ldots,\ell_n$ applied to each asset, and connected simply through quadratic covariation $\rho_{j,j}$ between the Brownian motions driving assets $i$ and $j$.
This is equivalent to the following $n$-dimensional diffusion process where we keep the $W$'s independent and where we embedded Brownian covariation into the diffusion matrix $\tilde{C}$, whose $i$-th row we denote by $\widetilde{C}_i$:
\begin{equation}\label{edsSCMD1}
d\underline{S}(t) = diag(\underline{\mu}) \underline{S}(t)dt +
diag(\underline{S}(t))
\widetilde{{C}}(t,\underline{S}(t))d\underline{W}(t), \ \ \  \tilde{a}_{i,j}(t,\underline{S}) := \widetilde{{C}}_{i} \widetilde{{C}}_{j}^T
\end{equation}
\begin{equation}
\tilde{a}_{i,j}(t,\underline{S}) = s_i(t,S_i) s_j(t,S_j)  \rho_{ij}= \left(\frac{\sum_{k=1}^N
\lambda_i^{k} \sigma_i^k(t)^2
\ell_{i,t}^k(S_i)}{\sum_{k=1}^N \lambda_{i}^{k}
\ell_{i,t}^k(S_i)} \ \  \frac{\sum_{k=1}^N \lambda_j^k
\sigma_{j}^k(t)^2 \ell_{j,t}^k(S_j)}{\sum_{k=1}^N
\lambda_{j}^k \ell_{j,t}^k(S_j)}\right)^{1/2} \rho_{ij}. \label{factorVol1}
\end{equation}
where $T$ represents the transposition operator.
\end{definition}

\noindent{\bf Assumption.} Throughout the paper we assume $\rho$ to be positive definite.

\begin{remark}\label{rem:scmdnomix} {\bf SCMD: no multivariate mixture}. It is important to point out  in SCMD that while
single--assets probability densities are mixtures by construction, the multivariate density is not a mixture of multivariate basic densities. The mixture property does not extend from the mono-dimensional dynamics to the multidimensional one.
\end{remark}
The practical use of the SCMD model is related to the following consideration.
Most often, one realistic way to price a plain
European option depending on more than one asset, especially in large dimension, is to use a Monte
Carlo simulation that samples suitably discretized paths according
to the drift rate of each component (risk--free minus dividend
yield) and to the diffusion matrix given by the local volatility
function in the mixture of densities model. Therefore,
assuming to have an exogenously computed structure of instantaneous
correlations $\rho_{ij}$ (computed e.g. through historical analysis
or implied by market instruments and supposed constant over time)
among the assets' returns, we could apply a na\"{\i}ve Euler Monte
Carlo scheme and simulate the joint evolution of the assets through
a suitably discretized time grid $\tau_1=0 \cdots \tau_N=T$ with a
covariance matrix whose $(i,j)$ component over the $(\tau_l,
\tau_{l+1})$ propagation interval is given by \eqref{factorVol1} computed at $t=\tau_l$.
It is immediate by construction that the SCMD approach is consistent with both the individual dynamics induced
by a LMD model for each underlying asset and with the imposed "instantaneous
correlation" (Brownian quadratic covariation) structure $\rho_{ij}$.

However, besides the practical
possibility of controlling the instantaneous correlation, and that
the number of base univariate densities to mix does not increase
with the number of underlying assets, one must be aware of the SCMD main
limitations, especially the following one.
%
For European type basket options we do not really need the full dynamics when it comes to actually computing the price, even with several maturities in the picture. Indeed, for each maturity $T$ the
payout depends only upon the values of the assets at
time $T$, i.e., upon the values $S_k(T)$, $\forall k$, regardless of
the history of prices. So in order to compute the risk--neutral
expectation in (\ref{European_option_price}) giving the price $\Pi$,
the only information we need is the joint density of the process
$(S_1(T), S_2(T), \ldots, S_n(T))$ of random variables under that particular
risk--neutral measure. This density is usually called the state price
density. In SCMD we do not know this density, so we have to generate samples
from the entire path $B_t$ for $0 \leq t \leq T$ . The
discretization time steps $\tau_{l+1} - \tau_l$ should be chosen
carefully to be sure that the numerical scheme used to generate the
discrete samples produces reasonable approximations. Notice that when
the maturity $T$ increases, more time steps are needed. This is particularly relevant in calibrating the model for risk management applications, for example, where the inverse problem can become daunting if the dimension is large and the discretization step small.

\subsection{The Multivariate Mixture Dynamics approach}

One could try to do something different and approach
the problem so that, under suitable assumptions, the individual LMD
models (one for each underlying asset, separately calibrated each on its
volatility surface) could be merged so as to provide a coherent
multi-asset model that allows for a degree of (semi)analytic
tractability comparable to the one typical of the univariate case. This will lead to a model where the mixture property is lifted to the multivariate density, contrary to the SCMD case
(Remark \ref{rem:scmdnomix} above).

%
%
%
%
Consider an $n$--dimensional stochastic process ${\underline{S}}(t)
= \left[S_1(t),\cdots,S_n(t)\right]^T$ whose generic $i-$ th
component follows the SDE
\begin{equation}
d S_i(t) = \mu_i S_i(t) dt + S_i(t) {C}_i(t,\underline{S})
d\underline{W}(t)
\end{equation}
where $\mu_i$ is a constant, $\underline{W} =
\left[W_1,\cdots,W_n\right]^T$ is a standard $n$--dimensional
Brownian motion and
${C}_i(t,\underline{S})$ is a row vector whose components are
deterministic functions of time and of the state of the process
$\underline{S}$.

Denote $a_{ij}(t,{\underline{S}}) = {C}_i(t,\underline{S})\
{C}_j^T(t,\underline{S})$. The associated Kolmogorov forward PDE to be
satisfied by the probability density $p_{\underline{S}(t)}$ of the
stochastic process $\underline{S}$ is
\begin{equation}
\frac{\partial p_{\underline{S}(t)}}{\partial t} + \sum_{i=1}^n
\frac{\partial}{\partial x_i}[\mu_i x_i p_{\underline{S}(t)}] -
\frac{1}{2} \sum_{i,j=1}^n \frac{\partial^2}{\partial x_i \partial
x_j}[ a_{ij} x_i x_j p_{\underline{S}(t)}]=0 \label{kolmogorov}
\end{equation}
where all functions are evaluated at $(t,\underline{x})$ for all $t
\geq 0, \underline{x} \in \mathbb{R}^n.$

With this notation $\underline{S}$ is given by the SDE
\begin{equation}\label{edsMVMD}
d\underline{S}(t) = diag(\underline{\mu}) \underline{S}(t)dt +
diag(\underline{S}(t))
{C}(t,\underline{S}(t))d\underline{W}(t)
\end{equation}
where ${C}$ is the $n \times n$ matrix whose $i$ th row is
${C}_i$.

${C}$ must be chosen so as to grant a unique strong solution
to the SDE (\ref{edsMVMD}). In particular, ${C}$ is assumed
to lead to a locally Lipschitz $a(t,\underline{x})$ and to satisfy, for a suitable positive constant $K$, the generalized
linear growth conditions
\begin{equation}
\mbox{trace} (a(t,\underline{x})) \|{\underline{x}}\|^2 \le K
(1+\|{\underline{x}}\|^2).
\end{equation}
The symbol $\| \|$ denotes here vector and matrix norms.

\bigskip

Consider an $n$ dimensional stochastic process
$\underline{X}^{(k)}$ whose generic $i$ th component follows the
dynamics
\begin{equation}\label{eq:instruX}
d X_i^{(k)}(t) = \mu_i X_i^{(k)}(t) dt + X_i^{(k)}(t)
\overline{\sigma}_i^{(k)}(t,\underline{X}^{(k)}) d\underline{W}(t)
\end{equation}
with $\overline{\sigma}_i^{(k)}(t,\underline{X}^{(k)})$ an $1\times n$
matrix satisfying particular conditions ensuring that the resulting
SDE giving the dynamic of $\underline{X}^{(k)}$ has a unique strong
solution.

Denote $a_{ij}^{(k)}(t,{\underline{X}^{(k)}}) =
\overline{\sigma}_i^{(k)}(t,\underline{X}^{(k)})\
\overline{\sigma}_j^{(k)}(t,\underline{X}^{(k)})^T$ and $p_t^{(k)}$ the
probability density function of $\underline{X}^{(k)}.$ The
associated Kolmogorov equation to be satisfied by $p_t^{(k)}$ is
\begin{equation}
\frac{\partial p_t^{(k)} (\underline{x})}{\partial t} + \sum_{i=1}^n
\frac{\partial}{\partial x_i}[\mu_i x_i p_t^{(k)} (\underline{x})] -
\frac{1}{2} \sum_{i,j=1}^n \frac{\partial^2}{\partial x_i \partial
x_j}[ a_{ij}^{(k)}(t,{\underline{x}}) x_i x_j p_t^{(k)}
(\underline{x})]=0. \label{kolmogorov_k}
\end{equation}

\bigskip

Inspired by the univariate approach which gave rise to the LMD model,
let us postulate that the density at any time $t$ of the
multivariate process $\underline{S}$ be equal to a weighted average
of the $p_t^{(k)}$

\begin{equation}
p_{\underline{S}(t)}({\underline{x}})=\sum_{k=1}^N \lambda^k
p_t^{(k)}({\underline{x}}), \hspace{2.cm} \lambda^k \ge 0 \ \forall
k, \hspace{2.cm} \sum_{k=1}^N \lambda^k=1 \label{mixture}
\end{equation}

The condition that $p_{\underline{S}(t)}$ satisfy
Eq.(\ref{kolmogorov}) and that each $p_t^{(k)}$ satisfy the equation
(\ref{kolmogorov_k}) leads through standard algebra to the PDE
\begin{equation}\label{eq:pdemixC}
\frac{1}{2} \sum_{i,j=1}^n \frac{\partial^2}{\partial x_i \partial
x_j} \left[ \left( a_{ij}(t,{\underline{x}}) \sum_{k=1}^N \lambda^k
p_t^{(k)}(\underline{x}) - \sum_{k=1}^N \lambda^k
a_{ij}^{(k)}(t,{\underline{x}}) p_t^{(k)}(\underline{x})
\right) x_i x_j \right] =0.
\end{equation}

\begin{proposition}
The unique candidate solution of the PDE \eqref{eq:pdemixC} is
\begin{equation}
a_{ij}(t,{\underline{x}}) = \frac{\sum_{k=1}^N \lambda^k
a_{ij}^{(k)}(t,{\underline{x}})
p_t^{(k)}(\underline{x})}{\sum_{k=1}^n \lambda^k
p_t^{(k)}(\underline{x})}, \ \ \  a_{ij}^{(k)}(t,{\underline{x}}) =
\overline{\sigma}_i^{(k)}(t,\underline{x})\
\overline{\sigma}_j^{(k)}(t,\underline{x})^T. \label{choice}
\end{equation}
\end{proposition}
\begin{proof}
It can be easily proven that the most general
solution of the equation $\sum_{ij} \frac{\partial^2}{\partial x_i
\partial x_j} f_{ij}({\underline{x}})=0$ has a Fourier transform
satisfying $({\bf q},f({\bf q}) {\bf q})=0$. The only matrix
function $f({\bf q})$ satisfying it and infinitely differentiable
with respect to ${\bf q}$ is constant. This constant must be zero in
order to have finite first and second moments of the multivariate
density $p_{\underline{S}(t)}$.
\end{proof}

This leads to the following definition.
\begin{definition} The general Multivariate Mixture Dynamics (MVMD) candidate Model for the vector of asset prices  $\underline{S}$ is defined as given by Equations
\eqref{edsMVMD} and \eqref{choice}. If a unique solution for the model equations exists and it admits a multivariate probability density, this is a mixture of basic multivariate densities according to Eq
\eqref{mixture}, where each $p^k$ is a multivariate basic density associated with an instrumental multivariate diffusion process \eqref{eq:instruX}.
\end{definition}
\begin{remark} {\bf MVMD: multivariate mixture}.
MVMD has been designed so as to have a mixture law for the multivariate model, contrary to SCMD, see Remark \ref{rem:scmdnomix} above.
\end{remark}

Of course to show that this is indeed a model we need to prove that the equation has a unique solution. We thus specialize our framework to a fully tractable case.

\subsection{The lognormal case and the univariate - multivariate MD connection}

We now specialize our framework by assuming that  the volatility
coefficient matrix for the $k$--th "base" density $p_t^{(k)}$ of Eq.
(\ref{kolmogorov_k}) is a deterministic function of time,
independent of the state, and of the particular form
$a_{ij}^{(k)}(t,\underline{x})= \overline{\sigma}_i^{(k)}(t)\
\overline{\sigma}_j^{(k)}(t)^T$.

Under this hypothesis we already know the dynamics corresponding to
Eq. (\ref{kolmogorov_k}), since we are dealing with multivariate geometric Brownian motions for \eqref{eq:instruX},
and we can explicitly write their densities
$p_t^{(k)}$
\begin{equation}
p_t^{(k)}(\underline{x}) = \frac{1}{(2 \pi)^{\frac{n}{2}} \sqrt{\det
\Xi^{(k)}(t)} \Pi_{i=1}^n x_i} \exp\left[ -\frac{ \tilde{x}^T
(\Xi^{(k)}(t))^{-1} \tilde{x}}{2} \right],
\label{multivariatelognormal}
\end{equation}
where $\Xi^{(k)}(t)$ is the $n \times n$ integrated covariance
matrix of returns for the many components of the process
${\underline{X}^{(k)}}$:
\begin{equation}
\Xi_{ij}^{(k)}(t) = \int_0^t \overline{\sigma}_i^{(k)}(s) \overline{\sigma}_j^{(k)}(s)^T ds
\end{equation}
($\Xi^{(k)}$ is assumed to be invertible at all times and
instantaneous correlation is included into the vector components) and
\begin{equation}
\tilde{x}_i^{(k)}=\ln x_i - \ln x_i(0) - \mu_i t + \int_0^t
\frac{\overline{\sigma}_i^{(k)^2}(s)}{2} ds. \label{argGauss}
\end{equation}

Calculations are simpler under the further
assumption that instantaneous correlation is constant in time, namely
\[ \overline{\sigma}_i^{(k)}(t) \overline{\sigma}_j^{(k)}(t)^T = \|\overline{\sigma}_i^{(k)}(t)\| \|\overline{\sigma}_j^{(k)}(t)\|\rho_{ij} =: \sigma_i^k(t)  \sigma_j^k(t)  \rho_{i,j},\]
or in other terms, assuming that 
\begin{equation}\label{eq:choleskysigk} \rho = B B^T, \ \  \overline{\sigma}_i^{(k)}(t) := \sigma_i^{(k)}(t) B_i 
\end{equation}
via diagonalization or Cholesky decomposition and  for positive and regular scalar time functions $\sigma_i^{(k)}(t)$, 
where $B_i$ is the $i$-th row of $B$.
The fact that the densities will get mixed up through
Eq. (\ref{mixture}) will have important consequences on the actual
structure of correlations, both instantaneous and average. But
first, let us prove that under a further assumption we can be fully
consistent with the dynamics specified by the LMD model for the
individual assets.

Let's assume that we have calibrated an LMD model for each $S_i(t)$:
if $p_{S_i(t)}$ is the density of $S_i$, we write
\begin{equation}
p_{S_i(t)}(x) = \sum_{k=1}^{N_i} \lambda_{i}^{k} \ell_{i,t}^k(x),
\hspace{.2cm} \mbox{with} \hspace{.2cm} \lambda_{i}^k \ge 0, \forall
k \hspace{.2cm} \mbox{and} \hspace{.2cm} \sum_k \lambda_{i}^k = 1
\label{individualDensity}
\end{equation}
where $Y_i^{1},...,Y_i^{N_i}$ are instrumental processes for $S_i$
evolving lognormally according to the stochastic differential
equation:
\begin{equation}\label{edsSi_k}
dY_i^{k}(t) = \mu_i Y_i^{k}(t) dt + \sigma_i^{k}(t) Y_i^{k}(t)
dZ_i(t),\ \ \ d \langle Z_i, Z_j\rangle_t = \rho_{ij} dt
\end{equation}
with density $\ell_{i,t}^k$.

For notational simplicity we will assume that the number of base
densities $N_i$ will be the same, $N$, for all assets. The exogenous
correlation structure $\rho_{ij}$ is given by the symmetric,
positive--definite matrix $\rho$.

The most natural tentative choice for the base densities of Eq.
(\ref{mixture}) is
\begin{equation}
p_{\underline{S}(t)}(\underline{x})=
\sum_{k_1,k_2,\cdots k_n=1}^N
\lambda_1^{k_1} \cdots \lambda_n^{k_n}
\ell_{1,\ldots,n; t}^{k_1,\ldots,k_n}(\underline{x}), \ \ \ell_{1,\ldots,n;t}^{k_1,\ldots,k_n}(\underline{x}) := p_{\left[Y_1^{k_1}(t),...,Y_n^{k_n}(t)\right]^T}(\underline{x}),
\label{mixture_k}
\end{equation}
or more explicitly
\begin{eqnarray}
\ell_{1,\ldots,n; t}^{k_1,\ldots,k_n}(\underline{x}) =
\frac{1}{(2 \pi)^{\frac{n}{2}} \sqrt{\det \Xi^{(k_1 \cdots k_n)}(t)}
\Pi_{i=1}^n x_i} & & \exp\left[ -\frac{ \tilde{x}^{(k_1 \cdots k_n)
T} \Xi^{(k_1 \cdots k_n)}(t)^{-1} \tilde{x}^{(k_1 \cdots k_n)} }{2}
\right]. \nonumber\label{multivariateGaussian}
\end{eqnarray}

Here, $\Xi^{(k_1 \cdots k_n)}(t)$ is the integrated covariance
matrix whose $(i,j)$ element is
\begin{equation}
\Xi_{ij}^{(k_1 \cdots k_n)}(t) = \int_0^t \sigma_i^{k_i}(s)
\sigma_j^{k_j}(s) \rho_{ij} ds \label{covarianceMatrix}
\end{equation}
and, generalizing Eq. (\ref{argGauss})
\begin{equation}\label{xtild_dimn}
\tilde{x}_i^{(k_1 \cdots k_n)}=\ln
 x_i - \ln x_i(0) - \mu_i t + \int_0^t \frac{\sigma_i^{k_i^2}(s)}{2} ds.
\end{equation}

\bigskip

Then, specializing (\ref{choice}), we have the following

\begin{definition}
The multivariate extension of the LMD model that we call Lognormal Multi Variate Mixture Dynamics (LMVMD) model is given by Eqs. (\ref{edsMVMD}) and (\ref{choice})   under specification (\ref{eq:choleskysigk}), leading to
\begin{eqnarray}\label{edsMVMDwithindBM1}
d\underline{S}(t) &=&  diag(\underline{\mu})\ \underline{S}(t)\ dt +
diag(\underline{S}(t))\
{C}(t,\underline{S}(t)) {B}  \ d\underline{W}(t), \\  \nonumber
 C_i(t,\underline{x}) &:=&\frac{\sum_{k_1,...,k_n=1}^{N} \lambda_1^{k_1}...\lambda_n^{k_n}\ \sigma_i^{k_i}(t) B_i \
\ell_{1,\ldots,n;t}^{k_1,\ldots,k_n}(\underline{x})}{\sum_{k_1,...,k_n=1}^{N}
\lambda_1^{k_1}...\lambda_n^{k_n}\
\ell_{1,\ldots,n;t}^{k_1,\ldots,k_n}(\underline{x})}
\end{eqnarray}
and therefore, defining consistently with earlier notation $a = C B (C B)^T$, 
\begin{equation}\label{C}
a(t,\underline{x}) 
= \frac{\sum_{k_1,...,k_n=1}^{N} \lambda_1^{k_1}...\lambda_n^{k_n}\
{V}^{k_1,...,k_n}(t)\
\ell_{1,\ldots,n;t}^{k_1,\ldots,k_n}(\underline{x})}{\sum_{k_1,...,k_n=1}^{N}
\lambda_1^{k_1}...\lambda_n^{k_n}\
\ell_{1,\ldots,n;t}^{k_1,\ldots,k_n}(\underline{x})}
\end{equation}
where
\begin{equation}\label{V}
{V}^{k_1,...,k_n}(t) = 
 \left[\sigma_i^{k_i}(t)\ \rho_{i,j}\
\sigma_j^{k_j}(t)\right]_{i,j = 1,...,n}.
\end{equation}

To avoid lengthy acronyms and with a slight abuse of notation, we will refer to the LMVMD model simply as MVMD, assuming implicitly from now on that we are dealing with the lognormal case.
\end{definition}

Putting notational complexity aside, what we ultimately did is to
mix in all possible ways the component densities for the individual
assets, still ensuring consistency with the starting models for the
components assets, and imposing the instantaneous correlation
structure $\rho$ at the level of the constituent densities.

To confirm that with MVMD we have a full model and not just a candidate model,
we need the following theorem.

\begin{theorem}
Under the assumption that the volatilities $\sigma_i^{k_i}(t)$ for all $i$ are once continuously differentiable and uniformly bounded from below and above by two positive real numbers  $\tilde{\sigma}$ and $\hat{\sigma}$ respectively, and that they take a common constant value $\sigma_0$ for $t \in [0,\epsilon]$ for a small positive real number $\epsilon$,  namely
\begin{eqnarray*}
&& \tilde{\sigma}=\inf_{t \ge 0} \left( \   \   \    \min_{i=1 \cdots n, k_i=1, \cdots
N} \ \ (
\sigma_i^{k_i}(t) ) \right) \\
&& \hat{\sigma}=\sup_{t \ge 0} \left( \ \ \  \max_{i=1 \cdots n, k_i=1 \cdots N} \ \ (
\sigma_i^{k_i}(t) ) \right) \\
&& \sigma_i^{k_i}(t)  = \sigma_0 > 0 \ \ \mbox{for all} \ \ t \in [0, \epsilon],
\end{eqnarray*}
and assuming the matrix $\rho$ to be positive definite, the MVMD n-dimensional stochastic differential equation (\ref{edsMVMDwithindBM1})
admits a unique strong solution. The diffusion matrix $a(t,\underline{x})$ in (\ref{C}) is positive definite for all $t$ and $x$.
\end{theorem}
\begin{proof}
Existence and uniqueness of a strong solution follows analogously to the univariate case \cite{mixture1}. Indeed, from (\ref{C}) we see that $a$ is a weighted average of $V$'s in (\ref{V}), with positive (and state-dependent) weights. Since we are assuming $\sigma$'s to be uniformly bounded and $\rho$'s to be positive definite, all $V$ matrices are positive definite and such is $a$. Moreover, all $a$'s entries are immediately seen to be bounded above and and below, the diagonal terms being bounded below by positive quantities.  Standard algebra yields
\begin{eqnarray*}
&& n \tilde{\sigma}^2 \le \|{a }\|^2=\sum_{i,j=1}^n
a_{ij}(t,\underline{S})^2 \le
n^2 \hat{\sigma}^2,
\end{eqnarray*}
so that we have a uniformly bounded continuously-differentiable function (hence locallly Lipschitz), and the usual linear growth condition holds. The common value $\sigma_0$ in an initial transient interval is needed to simplify the analysis of the component densities limit for $t \downarrow 0$ when the initial conditions for the single--asset densities are taken as Dirac delta functions.
\end{proof}

We now check that MVMD is indeed consistent with the mixture of densities models LMD
through which we have specified the dynamics of the single
components of $\underline{S}$ in the beginning.
\begin{proposition}
For any smooth test function $f : \mathbb{R} \longrightarrow \mathbb{R}$
and any $t \ge 0$,
the expectation of $f(S_i(t))$ is the same under the SCMD model   \eqref{edsSCMD1}, \eqref{factorVol1} and the MVMD model
\eqref{edsMVMD}, \eqref{C}.
%
%
%
\end{proposition}

\begin{proof}
The proof is trivial: let us start from the MVMD model. It is enough to compute the multiple integral
\begin{equation}
\begin{array}{ll}
\mathbb{E}_0\{f(S_i(t))\} & = \int dx_1 \cdots \int dx_i
\cdots \int dx_n f(x_i) p_{\underline{S}(t)}(\underline{x}) \\
&\\
= \sum_{k_1,k_2,\cdots k_n=1}^N &\lambda_1^{k_1}\cdots
\lambda_n^{k_n} \int dx_1\cdots\int dx_i \cdots \int dx_n f(x_i)
\ell_{1,\ldots,n;t}^{k_1,\ldots,k_n}(\underline{x})
\end{array}
\end{equation}

Integrating out all variables but $x_i$ in each of the integrals in
the right hand side we have
\begin{equation}
\begin{array}{ll}
\mathbb{E}_0\{f(S_i(t)) \} & = \sum_{k_1,k_2,\cdots k_n=1}^N
\lambda_1^{k_1} \cdots \lambda_n^{k_n} \int dx_i
f(x_i) \ell_{i,t}^{k_i}(x_i)\\
&\\
& = \sum_{k_i=1}^N \lambda_i^{k_i} \int dx_i f(x_i)
\ell_{i,t}^{k_i}(x_i) = \int dx_i f(x_i) p_{S_i(t)}(x_i)
\end{array}
\end{equation}
where the last integral is the same under SCMD, since by the condition that probability integrate up to one we know
that $\sum_{k=1}^N \lambda_i^k = 1$ for all $i$.
\end{proof}

\subsection{Dimensionality issues}

The computational scheme shown above ensures full consistency
between the single--asset and the multi--asset formulations of the
mixture of lognormal densities' model. It must be borne in mind,
however, that the number of ``base'' multivariate densities of the
formulation of Eqs. (\ref{mixture_k})--(\ref{xtild_dimn}) explodes
as $N^n$ if we have $N$ base univariate densities for each of the
$n$ underlying assets (more generally, if asset $i$ relies on a
single--asset mixture theory based on $N_i$ densities, the number of
multivariate densities entering the superposition amounts to
$\prod_{i=1}^n N_i$). This ``combinatorial explosion'' seems to
limit the applicability of the theory to baskets made of very few
assets.

However, as already observed elsewhere \cite{mixtureFX} two
empirical facts appear in the univariate mixture of densities model,
that encourage the application of the model to real world
multivariate settings. They are briefly summed up here:
\begin{itemize}
\item the number of base densities $N$ needed to reproduce
accurately enough the implied volatility surface for a single asset
is typically 2 to 3;
\item there appears to be a clear hierarchy between densities
composing the mixture, dictated by the weights $\lambda_k$ borne by
each density in the superposition (\ref{mixtureOne}): in fact, typically one density
takes up most of the weight, the second takes up most of the
remaining weight (remember that $\sum_{k=1}^{N} \lambda_{k}=1$) and
the last weighs little compared to the first two.
\end{itemize}

The consequences of the first issue are evident: the base in the
power law $N^n$ is of the order of two/three. This is not enough to
completely solve the explosion problem: taking $N=3$ and $n=8$ still
implies that in order to compute the price of an European option on
the basket, we should compute 6561 multidimensional integrals.

However, the second point ensures that most of the multivariate
coefficients $\lambda_1^{k_1} \cdots \lambda_n^{k_n}$ result from
the product of the smallest $\lambda$, thus rendering the
corresponding terms in the expansion of Eq. (\ref{mixture_k})
negligible. Given any $0\le\kappa\le1$, a possible solution can
therefore be to approximate Eq. (\ref{mixture_k}) through
\begin{eqnarray}
p_{\underline{S}(t)}(\underline{x}) \simeq
p_{\underline{S}(t)}(\underline{x},\kappa) = \sum_{ 
(k_1, \ldots, k_n) \in {\cal I}(\lambda)} \left(
\prod_{j=1}^n \tilde{\lambda}_j^{k_j} \right)
\ell_{1,\ldots,n;t}^{k_1,\ldots,k_n}(\underline{x}), \\ \nonumber
\ {\cal I}(\lambda) := \{(k_1,k_2,\ldots,k_n): \ \ k_i \in \mathbb{N} \cap [1,N],\  \prod_{j=1}^n \lambda_j^{k_j} > \kappa \}
\label{mixture_k_cutoff}
\end{eqnarray}
$\kappa$ therefore plays the role of a "cutoff parameter" that
ensures that only significant contributions to the multivariate
expansion are retained; in order to preserve normalization of the
resulting density,
\begin{equation}
\prod_{j=1}^n \tilde{\lambda}_j^{k_j}= \frac{
  \prod_{j=1}^n \lambda_j^{k_j}
} {
  \sum_{(k_1, \ldots, k_n) \in {\cal I}(\lambda)}
  \prod_{j=1}^n \lambda_j^{k_j}
}.
\end{equation}

Note that $p_{\underline{S}(t)}(\underline{x}) =
p_{\underline{S}(t)}(\underline{x},\kappa=0)$, whereas increasing
$\kappa$ decreases the number of base multivariate densities in the
approximate expansion of Eq. (\ref{mixture_k_cutoff}); $\kappa$
therefore controls the tradeoff between the accuracy in the
approximation and the computational efficiency.

In order to have an estimate of the computational gain due to a
choice of $\kappa \ne 0$, we can compute how the volume in
$n$--dimensional space of the region $\kappa < \prod_{i=1}^n x_i \le
1$ scales for fixed cutoff as a function of $n \ge 1$: the recursive
law is
\begin{equation}
V_n(\kappa)=\int_{\kappa}^1 dx_1 \int_{\frac{\kappa}{x_1}}^1 dx_2
\cdots \int_{\frac{\kappa}{\prod_{i=1}^{n-1} x_i}}^1 dx_n =
V_{n-1}(\kappa) + \frac{(-1)^n}{(n-1)!} \, \kappa \, (\ln
\kappa)^{n-1}
\end{equation}
setting $V_0=1$ conventionally.

Now, let us neglect the striking feature that there exists a strong
hierarchy between components in the univariate mixtures of densities
(point two above). Suppose instead that we are in a less favorable
case, namely that the density of coefficients $\lambda$ of the
mixture model for each asset is uniform and equal to $\rho$ ({\it
i.e.} the distance on the $[0,1]$ interval between consecutive
$\lambda$ is equal to $\frac{1}{\rho}$); then, the density of
coefficients in the multivariate theory is $\rho^n$. An estimate of
the number of multivariate densities involved in the expansion of
Eq. (\ref{mixture_k_cutoff}) is $\mathfrak{N}_n(\kappa)=V_n(\kappa)
\rho^n$. To give an example, if $\kappa=5\%$ and $\rho=3$, the
number of densities has a maximum at $\mathfrak{N}_n(5\%) \simeq 80$
for $n \simeq 8$: neglecting the densities that contribute to 5\% of
the normalization already yields much less than the full 6581 set of
eight variate densities.

\section{Analysis and comparison of dependence structures}\label{sec:dependencest}

\subsection{Instantaneous correlations in the SCMD and MVMD models}

From \eqref{factorVol1} and \eqref{C}-\eqref{V} we can compare the expression for the instantaneous local covariance, or quadratic covariation, between asset returns in the two models, SCMD and MVMD. Without loss of generality consider a two--dimensional
process, namely take $n=2$. To lighten notation we omit the time argument in volatilities $\sigma_i^{k}(t)$. 

Recall that in a SCMD scheme the instantaneous
variance for the log $S_1$ asset, say, at time $t$  would be (see Eq. (\ref{factorVol1}))
\begin{equation}
\tilde{C}_{11}(x_1,t) =
\frac{\sum_{k=1}^N {\lambda_1}^k \sigma_1^{(k)2}
\ell_t^{(1k)}(x_1)}{\sum_{k=1}^N {\lambda_1}^k
\ell_t^{(1k)}(x_1)}
\label{tildeC11}
\end{equation}
and the instantaneous covariance of returns, or quadratic covariation, between the two assets would be
\begin{equation}
\tilde{C}_{12}(x_1,x_2,t) =
\sqrt{\frac{\sum_{k=1}^N {\lambda_1}^k \sigma_1^{(k)}
\ell_t^{(1k)}(x_1)}{\sum_{k=1}^N {\lambda_1}^k
\ell_t^{(1k)}(x_1)}} \sqrt{\frac{\sum_{k=1}^N {\lambda_2}^k
\sigma_2^{(k')} \ell_t^{(2k)}(x_2)}{\sum_{k=1}^N {\lambda_2}^k
\ell_t^{(2k)}(x_2)}} \rho
\label{tildeC12}
\end{equation}
to be compared with the expressions
\begin{equation}
C_{11}(x,x_2,t) = \frac{\sum_{k,k'=1}^N {\lambda_1}^k {\lambda_2}^{k'}
\sigma_{1}^{(k)2} \ell_{t}^{(kk')}(x_1,x_2)}{\sum_{k,k'=1}^N
{\lambda_1}^k {\lambda_2}^{k'} \ell_{t}^{(kk')}(x_1,x_2)}
\label{C11}
\end{equation}
and
\begin{equation}
C_{12}({\bf x},t) = \frac{\sum_{k,k'=1}^{\nu} {\lambda_1}^k {\lambda_2}^{k'}
\sigma_1^{(k)} \sigma_2^{(k')} \rho \,
\ell_t^{(kk')}(x_1,x_2)}{\sum_{k,k'=1}^{\nu} {\lambda_1}^k {\lambda_2}^{k'}
\ell_t^{(kk')}(x_1,x_2)}
\label{C12}
\end{equation}
of MVMD.

An evident difference is that, while in Eq. (\ref{tildeC11}) the
instantaneous covariance of log $S_1$ depends only on $x_1$ itself, and not on
$x_2$, the opposite is true of Eq. (\ref{C11}). In other words, the
diffusion matrix is now \emph{fully} dependent on the components of
the multidimensional process. Moreover, the two equations
(\ref{tildeC12}) and (\ref{C12}) are structurally different. However,
there must be a link between the two: we know that in the
limit when the correlation $\rho$ between variables $\ln(S_1)$ and $\ln(S_2)$ vanishes,
they will in fact evolve ignoring one another in both models.

By the choice we made at the beginning, $\ell_t^{(kk')}$ is a bivariate
lognormal density, {\it i.e.} it has the expression
\begin{equation}
\begin{array}{r}
\ell_t^{(kk')}(x_1,x_2) = \frac{1}{2 \pi \sqrt{{\alpha}_{11} {\alpha}_{22}-\rho^2 {\alpha}_{12}^2}}
\frac{1}{x_1 x_2} \exp\left[ -\frac{1}{2}
\tilde{x_1}^2 \frac{{\alpha}_{22}}{{\alpha}_{11}{\alpha}_{22}-\rho^2 {\alpha}_{12}^2}
-2 \tilde{x_1} \tilde{x_2} \frac{{\alpha}_{12} \rho}{{\alpha}_{11}{\alpha}_{22}-\rho^2 {\alpha}_{12}^2}
\right. \\
\\
\left.
+\tilde{x_2}^2 \frac{{\alpha}_{11}}{{\alpha}_{11}{\alpha}_{22}-\rho^2 {\alpha}_{12}^2}
\right]
\end{array}
\label{bivariate}
\end{equation}
with $\tilde{x_1}$ and $\tilde{x_2}$ defined as in Eq. (\ref{argGauss}) and
\begin{equation}
\left\{
\begin{array}{l}
{\alpha}_{11}=\int_0^t \sigma_1^{(k)^2}(s) ds \\
\\
{\alpha}_{22}=\int_0^t \sigma_2^{(k')^2}(s) ds \\
\\
{\alpha}_{12}=\int_0^t \sigma_1^{(k)}(s) \sigma_2^{(k')}(s) ds. \\
\end{array}
\right.
\end{equation}

The tetrachoric expansion for the bivariate normal density with
correlation $\rho$ reads \cite{Vasicek}
\begin{equation}
n(x_1,x_2,\rho)=n(x_1) n(x_2) \sum_{k=0}^\infty \frac{\rho^k}{k!} H_k(x_1) H_k(x_2)
\label{tetrachoric}
\end{equation}
($H_k$ is the k$^{th}$ Hermite polynomial); this, applied to
$\ell_t^{(kk')}$ yields
\begin{equation}
\begin{array}{ll}
\ell_t^{(kk')}(x_1,x_2) \simeq &
\frac{1}{\sqrt{2 \pi {\alpha}_{11}}}
\frac{1}{x_1} \exp\left[ -\frac{1}{2 {\alpha}_{11}} \tilde{x_1}^2 \right]
\frac{1}{\sqrt{2 \pi {\alpha}_{22}}}
\frac{1}{x_2} \exp\left[ -\frac{1}{2 {\alpha}_{22}}
\tilde{x_2}^2 \right] \\
\\
&
+
\frac{1}{\sqrt{2 \pi {\alpha}_{11}}}
\frac{1}{x_1} \exp\left[ -\frac{1}{2 {\alpha}_{11}} \tilde{x_1}^2 \right]
\frac{1}{\sqrt{2 \pi {\alpha}_{22}}}
\frac{1}{x_2} \exp\left[ -\frac{1}{2 {\alpha}_{22}}
\tilde{x_2}^2 \right]
\tilde{x_1} \tilde{x_2} \frac{{\alpha}_{12}}{{\alpha}_{11} {\alpha}_{22}}\rho + O(\rho^2)
\end{array}
\end{equation}
and similarly expanding Eqs. (\ref{C11}) and (\ref{C12}) we get the following
\begin{proposition} {\bf SCMD approximates MVMD for weakly correlated
systems.}
The SMCD and MVMD instantaneous covariance structures, or quadratic co-variations, coincide first order in the Brownian correlation $\rho$, namely 
\begin{equation}
\left\{
\begin{array}{l}
C_{11}(x_1,x_2,t) = \tilde{C}_{11}(x_1,t) + O(\rho^2) \\
\\
C_{12}(x_1,x_2,t) = \tilde{C}_{12}(x_1,x_2,t) + O(\rho^2). \\
\end{array}
\right.
\end{equation}
\end{proposition}
We also have the following 
\begin{corollary}\label{co:loccorr}{\bf Local correlation structure in MVMD and SCMD.}
We may define the instantaneous local correlation in a bivariate diffusion model as
\[ \rho_L(t) := \frac{d \langle S_1, S_2 \rangle_t}{\sqrt{d \langle S_1, S_1 \rangle_t \ d \langle S_2, S_2 \rangle_t } } .\] 
The instantaneous local correlation structure for SCMD is obviously the constant Brownian correlation $\rho_L(t) = \rho$, whereas for MVMD we have a smaller local correlation, in absolute value, given by
\begin{equation}
\rho_L(t) = \frac{
\rho \sum_{k,k'=1}^{\nu} {\lambda_1}^k {\lambda_2}^{k'} \sigma_1^{(k)}
\sigma_2^{(k')} \ell_t^{(kk')}(x_1,x_2)
}
{
\sqrt{
\left(
\sum_{k,k'=1}^{\nu} {\lambda_1}^k {\lambda_2}^{k'}
\sigma_1^{(k)2} \ell_t^{(kk')}(x_1,x_2)
\right)
\left(
\sum_{k,k'=1}^{\nu} {\lambda_1}^k {\lambda_2}^{k'}
\sigma_2^{(k')2} \ell_t^{(kk')}(x_1,x_2)
\right)
}
} \le \rho
\end{equation}
where the inequality follows from Schwartz's inequality.
\end{corollary}

\subsection{Terminal correlation}

In both SCMD and MVMD the log--return expectation for component $\ln S_1$ in $x_1$ is
\begin{equation}
\mathbb{E}_0\{ \ln (S_1(t)/S_1(0)) \} = \sum_k {\lambda_1}^k \int_0^t
 \left( \mu_s - \frac{\sigma_s^{(k)2}}{2} \right) ds
\end{equation}
and its variance is
\begin{equation}
\mbox{Var}_0\{ \ln(S_1(t)/S_1(0)) \} = \sum_k {\lambda_1}^k \int_0^t
\sigma_s^{(k)2} ds
\end{equation}
It is immediate to prove the following
\begin{proposition}{\bf Terminal covariance and correlations in MVMD.}
The terminal covariance of log--returns of $S_1$ and $S_2$ in SCMD is not known in closed form, whereas for MVMD we have 
\begin{equation}
\mbox{Cov}_0\{\ln(S_1(t)/S_1(0)) ,\ln(S_2(t)/S_2(0))  \} = \sum_{kk'}
{\lambda_1}^k {\lambda_2}^{k'} \rho \int_0^t \sigma_1^{(k)}(s) \sigma_2{(k')}(s)  ds,
\end{equation}
giving rise to a terminal correlation between returns up to time $t$ in MVMD given by
\begin{equation}
\hat{\rho}(t)  = \frac{\mbox{Cov}_0\left\{ \ln \frac{S_1(t)}{S_1(0)} ,\ln \frac{S_2(t)}{S_2(0)} \right\}}
{\sqrt{\mbox{Var}_0\left\{ \ln \frac{S_1(t)}{S_1(0)}\right\} \mbox{Var}_0\left\{ \ln \frac{S_2(t)}{S_2(0)}\right\}}}=
\frac{
\rho \sum_{kk'} {\lambda_1}^k {\lambda_2}^{k'} \int_0^t \sigma_1^{(k)}(s)
\sigma_2^{(k')}(s) ds
}
{
\sqrt{
(\sum_k {\lambda_1}^k \int_0^t \sigma_1^{(k)2}(s) ds)
(\sum_{k'} {\lambda_2}^{k'} \int_0^t \sigma_2^{(k')2}(s) ds)
}
}.
\end{equation}
\end{proposition}
In MVMD, the terminal correlation between assets is, again, a weighted average of terminal correlations in the various Black--Scholes states upon which the mixture is based.

Note that, instead, an analytical expression for terminal correlation does not exist for the SCMD, so that any comparison between the two must be done on a numerical basis. This is a further advantage of MVMD.

\subsection{Correlation between asset and local covariance}

It can be shown that the MVMD retains a property of single--asset mixture dynamics models regarding the terminal asset--variance correlation:
\begin{theorem}\label{th:loccovass}
Consider for all $i,j$ the random variable
$$
v_{ij}(T)=\frac{1}{T}\int_0^T a_{ij}(t,\underline{S}(t)) \, dt,
$$
$v(T)$ being the ''average percentage covariance'' of the process $\underline{S}$. Then for all $k$ 
$$
Corr_0\{a_{ij}(T,\underline{S}(T)),S_k(T)\}=0
$$
and
$$
Corr_0\{v_{ij}(T),S_k(T)\}=0
$$
for all $T$. At the same time, however, notice that $a_{ij}(T,\underline{S}(T))$ is a deterministic function of $\underline{S}$, whose components are all correlated with $S_k$.
In the univariate case $n=1$, in the LMD model (\ref{eq:dcmix}), one has the striking result 
\[Corr_0\{s(T,S(T))^2,S(T)\}=0, \ \  Corr_0\left\{\int_0^T s(t,S(t))^2 dt, \  S_1(T)\right\}=0, \]
with 
\[Corr_0\{d s(t,S(t)), dS(t)\} \in \{+1, -1\}. \]
showing that terminal correlation between assets and squared volatilities is zero despite the latter being deterministic functions of the former and thus instantaneously perfectly correlated. 
\end{theorem}

The proof follows closely \cite{general_mixture_diffusion}, and will be omitted here in the interest of brevity.

\begin{remark} {\bf Mixure Dynamics Models escape volatility-asset correlation criticism of common local volatility models}. It is worth repeating here a remark already made in \cite{general_mixture_diffusion}: despite the commonly cited drawback of local volatility models (the perfect instantaneous correlation between the asset and its local variance), mixture models feature {\it vanishing terminal correlations} between assets and average variances after $t=0$. This mitigates to some extent the  objection to local volatility models, at least only for the family of mixture  models.
\end{remark}

\subsection{Copula function in MVMD}\label{sect:copula_MVMD}

\begin{proposition}
The copula function associated to MVMD (\ref{edsMVMDwithindBM1}) can be written as
\begin{equation}
C (u_1,...,u_n) = \sum_{k_1,...,k_n = 1}^{N}
\lambda_1^{k_1}...\lambda_n^{k_n}\
\Phi_{M}\left(h_1(F_{S_1(t)}^{-1}(u_1)),...,h_n(F_{S_n(t)}^{-1}(u_n))\right)
\end{equation}
where $\Phi_M$ denotes the standardized $n$-dimensional normal
distribution function with correlation matrix ${M}$ given by
\begin{equation}\label{R}
{M}_{i,j} = \frac{\Xi_{ij}^{(k_1 \cdots k_n)}(t) }{\sqrt{\Xi_{ii}^{(k_1 \cdots k_n)}(t)\   \Xi_{jj}^{(k_1 \cdots k_n)}(t)}}
\end{equation}
for $i, j \in \{1,...,n\},$ where $\Xi$ has been defined in (\ref{covarianceMatrix}).
$ F_{S_i(t)}^{-1}$ is the
inverse of the cumulative distribution function $F_{S_i(t)}$ of
$S_i(t)$ given by :
\begin{equation}
F_{S_i(t)}(x) = \sum_{k_i = 1}^{N} \lambda_i^{k_i}\
\Phi\left[\frac{1}{\sqrt{\Xi_{ii}^{(k_1,\ldots,k_n)}(t)}}\left(\ln \frac{x} {Y_i^{k_i}(0)}-\mu_i
t+\frac{1}{2}\Xi_{ii}^{(k_1,\ldots,k_n)}(t)\right)\right],
\end{equation}
where $\Phi$ is the usual standard normal cumulative distribution
function, and $h_i$ is given by
%
\begin{equation}
h_i(x) = \frac{1}{V_i^{k_i}(t)}\left(\ln \frac{x}{Y_i^{k_i}(0)} -
\mu_i t+\frac{1}{2}\Xi_{ii}^{(k_1,\ldots,k_n)}(t)\right).
\end{equation}
\end{proposition}

\vspace{3mm}
\begin{proof}
We start by using the characterizaton (\ref{mixture_k}) of the multivariate law of MVMD, 
where each $Y_i^{k_i}$ evolves lognormally according to the SDE (\ref{edsSi_k}).
Using Corollary on page 47 of Nelsen \cite{nelsen}, we see that the MVMD copula is
\begin{equation}
C(u_1,...,u_n)=\sum_{k_1,...,k_n=1}^{N}\lambda_1^{k_1}...\lambda_n^{k_n}F_{[Y_1^{k_1}(t),...,Y_n^{k_n}(t)]^{T}}\left(F_{S_1(t)}^{-1}(u_1),...,F_{S_n(t)}^{-1}(u_n)\right)
\end{equation}
where $F_{[Y_1^{k_1}(t),...,Y_n^{k_n}(t)]^{T}}$ is the cumulative
distribution function of the vector
$[Y_1^{k_1}(t),...,Y_n^{k_n}(t)]^{T}.$

\bigskip

Because $F_{[Y_1^{k_1}(t),...,Y_n^{k_n}(t)]^{T}}$ is a cumulative
distribution function and each $Y_i^{k_i}$ evolve lognormally, it
follows that
\begin{equation}
\forall \left(x_1,...,x_n\right) \in \mathbb{R}^n, \
F_{[Y_1^{k_1}(t),...,Y_n^{k_n}(t)]^{T}}\left(x_1,...,x_n\right) =
\Phi_{M}\left(h_1(x_1),...,h_n(x_n)\right)
\end{equation}
from which the copula's expression follows.
\end{proof}

\bigskip

\begin{corollary} The MVMD copula is a mixture of multivariate copulas that are the
standardized multivariate normal distribution with correlation
matrix ${M}$ given by (\ref{R}) and marginals $G_1,...,G_n$
defined as follows :
\begin{equation} G_i (x) =
F_{S_i(t)}\left[\exp\left(\sqrt{\Xi_{ii}^{(k_1,\ldots,k_n)}(t)}\ x + \ln (Y_i^{k_i}(0))+\mu_i
t-\frac{1}{2}\Xi_{ii}^{(k_1,\ldots,k_n)}(t)\right)\right].
\end{equation}
\end{corollary}
\vspace{3mm}
\begin{proof}
$\Phi_{M}$ is the standardized $n$--dimensional normal
distribution function with correlation matrix ${M}$ given by
(\ref{R}).

We prove now that each
$\Phi_{M}\left(h_1(F_{S_1(t)}^{-1}(u_1)),...,h_n(F_{S_n(t)}^{-1}(u_n))\right)$
is a copula.

\bigskip

Because $\Phi_{M}$ is an $n$- dimensional distribution
function, we only need to prove that each $h_i \circ
F_{S_i(t)}^{-1}$ is inverse of a univariate distribution function
and then the result is deduced using Sklar's theorem. To this end
fix an $i \in \{1,\cdots,n\}$ and let $G_i$ be the function from
$\mathbb{R}$ to $[0,1]$ given by
$$G_i (x) = F_{S_i(t)}\left[\exp\left(\sqrt{\Xi_{ii}^{(k_1,\ldots,k_n)}(t)}\ x + \ln
(Y_i^{k_i}(0))+\mu_i t-\frac{1}{2}\Xi_{ii}^{(k_1,\ldots,k_n)}(t)\right)\right].$$
$G_i$ is a distribution function. Indeed, $G_i$ is increasing as composition of increasing
functions,  $\lim_{x \rightarrow -\infty} G_i(x) = \lim_{ x \rightarrow 0} F_{S_i(t)}(x) = 0$, $\lim_{x \rightarrow + \infty} G_i(x) = \lim_{x \rightarrow
+\infty} F_{S_i(t)}(x) = 1$, and  $G_i^{-1} = h_i \circ F_{S_i(t)}^{-1}.$
\end{proof}

\subsection{Rank correlations for normal mixtures}
We may also need a synthetic rank correlation measure for the statistical dependence between two assets returns, rather than the whole copula function. Indeed, we will use this quantity in our subsequent tests. To this end, we now compute Kendall's tau for a bivariate distribution that is a mixture of two bivariate normal distributions. The proof of the following proposition is straightforward.
\begin{proposition}\label{th:taumd}
Let us consider a bivariate random variable $(X,Y)$ defined as a
mixture of $2$ bivariate Gaussian random variables $(X_a,Y_a)$ and
$(X_b,Y_b)$.  $\lambda$ denotes the mixture coefficient. $\mu_{X_a}$
(resp. $\mu_{Y_a}$, $\mu_{X_b}$ and $\mu_{Y_b}$) denotes the mean of
$X_a$ (resp. $Y_a$, $X_b$ and $Y_b$). $\sigma_{X_a}$ (resp.
$\sigma_{Y_a}$, $\sigma_{X_b}$ and $\sigma_{Y_b}$) denotes the
standard deviation of $X_a$ (resp. $Y_a$, $X_b$ and $Y_b$). $\rho_a$
(resp. $\rho_b$) denotes the correlation between $X_a$ and $Y_a$
(resp. $X_b$ and $Y_b$).
\newline
Kendall's tau for $(X,Y)$ is given by:
\begin{eqnarray}\label{taumix}
\tau (X,Y) & = & \frac{2}{\pi} \left[\lambda^2  \arcsin
\left(\rho_a\right) + (1 - \lambda)^2  \arcsin
\left(\rho_b\right)\right] + 2 \lambda (\lambda - 1)\nonumber\\
 & + & 4 \lambda (1 -
\lambda)\left[ \Phi_{\rho} (m_X,m_Y) +  \Phi_{\rho} (-
m_X, - m_Y)\right],
\end{eqnarray}
where
$$m_X = \frac{\mu_{X_a} - \mu_{X_b}}{\sqrt{\sigma^2_{X_a} +
 \sigma^2_{X_b}}},\ \  m_Y = \frac{\mu_{Y_a} - \mu_{Y_b}}{\sqrt{\sigma^2_{Y_a} +
 \sigma^2_{Y_b}}},$$
$$\rho = \frac{\rho_b\ \sigma_{X_b}\ \sigma_{Y_b}}{\sqrt{\sigma^2_{X_a} +
 \sigma^2_{X_b}}\ \sqrt{\sigma^2_{Y_a} + \sigma^2_{Y_b}}} + \frac{\rho_a\ \sigma_{X_a}\ \sigma_{Y_a}}{\sqrt{\sigma^2_{X_a} +
 \sigma^2_{X_b}}\ \sqrt{\sigma^2_{Y_a} + \sigma^2_{Y_b}}}$$
and $\Phi_{\rho}$ is the cumulative distribution function of the bivariate normal variable with correlation coefficient $\rho$, with zero means and unit variances. 
\end{proposition}

\section{Markovian projections}\label{sect:Markovian_projection}

In this section we provide two Markovian projection results. First, we introduce a model that we call Multivariate Uncertain Volatility Model (MUVM) and we prove that MVMD is a
multivariate Markovian projection of MUVM. Secondly, we study the Markovian projection for the basket price process, deriving a consistency mixture result between the multivariate distribution and the basket dynamics for geometric baskets.

\subsection{MVMD as projection of MUVM}
We now introduce the Multivariate Uncertain Volatility Model (MUVM).
This is a model specified through a system of SDEs of the form
\begin{equation}\label{edsuncertain_volatily_model}
d\xi_i(t) = \mu_i\ \xi_i(t) dt + \sigma^{I_i}_i(t)\ \xi_i(t) dZ_i(t),\ \ i
= 1,...,n,
\end{equation}
where each $Z_i$ is a
standard one dimensional Brownian motion, $\mu_i$ are
constants, $\sigma^I:= [\sigma^{I_1}_1,\ldots,\sigma^{I_n}_n]^T$ is a random vector independent of $Z$ and representing uncertain volatilities. 
We assume that the assets
$\xi_i$ are pairwise correlated through the driving Brownian
motions covariation. To be more specific we assume that $d \left \langle Z_i,
Z_j \right \rangle_t = \rho_{i,j} dt$.
What is actually random in
$\sigma^I$ are the indices $I_1,\ldots,I_n$, in the different $\sigma^I$ components, each if which can take values $1,2,\ldots,N$ with different probabilities. $I_1,\ldots,I_n$ are assumed to be mutually independent.

More specifically,
each $\sigma^{I_i}_i$ takes values in a set of $N$ deterministic functions
$\sigma_i^k$ with probability $\lambda_i^k$ ($\sigma_i^{k}$ and
$\lambda_i^{k}$ as defined in the previous section). We thus have, for all times in $(\varepsilon,+\infty)$, with small $\varepsilon$,
\begin{equation*}
(t \longmapsto \sigma^{I_i}_i(t)) = \left\{\begin{array}{l} (t \longmapsto
\sigma_i^1(t))\ \  \mbox{with} \ \mathbb{Q} \ \ \mbox{probability}\ \  \lambda_i^1 \\
(t \longmapsto
\sigma_i^2(t))\ \   \mbox{with} \ \mathbb{Q} \ \ \mbox{probability}\ \  \lambda_i^2 \\
\vdots \\
(t \longmapsto \sigma_i^N(t))\ \  \mbox{with} \ \mathbb{Q} \ \ \mbox{probability}\ \
\lambda_i^N
\end{array} \right.
\end{equation*}
We assume that all the above volatilities for asset $i$ have a common time-path from 0 to $\epsilon/2$, and then from the reached common value $\sigma_i(\epsilon/2)$ at time $\epsilon/2$ each time-function connects  to the relevant $\sigma_i^k(\epsilon)$ to continue then as $\sigma_i^k$. This is an initial regularization that is needed to make the dynamics smooth and ensure existence and uniqueness of solutions for the related equation. If $\epsilon$ is small and the volatilities are smooth in time then we may neglect the initial regularization when computing expectations. We also assume that randomness of the time functions (or of the random indices $I$) is realized at time $\epsilon/2$. Hence the uncertainty is quite short-lived and after that every asset follows a geometric Brownian motion. For an analogous analysis of the univariate case and a discussion see \cite{general_mixture_diffusion}. 

\begin{remark}{\bf MUVM vs MVMD as financial models}. This is a good point to mention that the feature we just mentioned makes MUVM a quite stylized and debatable model, and indeed MVMD, whose link with MUVM we are going to clarify now, is definitely more interesting and well behaving. This is why we stress that MUVM is interesting both as a mathematical tool to originate MVMD and as a tool to clarify a number of features on dependence in  MVMD, but as pricing and hedging model per se MVMD remains superior in terms of smoothness, consistency and dynamics. While we support the use of MVMD, we do not recommend the use of MUVM as a standalone model. 
\end{remark}

\begin{lemma}
\label{gyongilemma} {\bf (Gy{\"o}ngi's Lemma \cite{gyongi})}.
Let us consider an $n$-dimensional stochastic process
$(\underline{\xi}_t)_{t\geq 0}$ starting from $0$ with the It{\^{o}} form:
\begin{equation}\label{eq:Gy{\"o}ngilem}
d\underline{\xi}_t = \underline{\beta}(t,\underline{\xi}_t) dt +
v(t,\underline{\xi}_t) d\underline{W}_t,
\end{equation}
where $\underline{W}$ is a standard $d$-dimensional Brownian motion,
$\underline{\beta}$ is an $n$-dimensional bounded process and $v$ is
an $n \times d$ bounded process with $vv^T$ being uniform positive
definite. There exists a Markovian $n$-dimensional process
$(\underline{X}_t)_{t\geq 0}$ which has the same distribution as $(\underline{\xi}_t)_{t\geq 0}$ at each fixed single time $t$, and which is a weak
solution to the following stochastic differential equation:
\begin{equation}
d\underline{X}_t = \underline{\mu}(t,\underline{X}_t) dt +
\sigma(t,\underline{X}_t) d\underline{W}_t, X_0 = 0,
\end{equation}
where
\begin{equation}
\sigma \sigma^T(t,\underline{x}) =
\mathbb{E}\left[vv^T|\underline{\xi}_t = \underline{x}\right],
\underline{\mu}(t,\underline{x}) =
\mathbb{E}\left[\underline{\beta}|\underline{\xi}_t =
\underline{x}\right].
\end{equation}
%
$\underline{X}$ is called the Markovian projection (in dimension $n$) of $\underline{\xi}$.
\end{lemma}

\vspace{3mm}

\begin{theorem}
The MVMD model is a Markovian projection in dimension $n$ of the MUVM.
\end{theorem}

\vspace{3mm}
\begin{proof}
The system of SDEs (\ref{edsuncertain_volatily_model}) can be
written in the following manner
\begin{equation}\label{edsscenariomodel} d\underline{\xi}(t) = diag(\underline{\mu})\ \underline{\xi}(t)\ dt + diag (\underline{\xi}(t))\ {A^I(t)}\
d\underline{W}(t)
\end{equation}
with $W$ a vector of $n$ independent
standard Brownian motions and ${A^I(t)}$ the Cholesky
decomposition of the covariance matrix
$\Sigma^I_{i,j}(t) := \sigma^{I_i}_i(t) \sigma^{I_j}_j(t)\ \rho_{ij}.$ This is our process \eqref{eq:Gy{\"o}ngilem} in Lemma~\ref{gyongilemma}.

\bigskip

The MVMD model given by (\ref{edsMVMDwithindBM1}) can
be written as
\begin{equation}\label{edsMVMDwithindBM}
d\underline{S}(t) =  diag(\underline{\mu})\ \underline{S}(t)\ dt +
\sigma (t,\underline{S}(t))\ d\underline{W}(t)
\end{equation}
where $\sigma(t,\underline{S}(t)) := diag(\underline{S}(t))\
{C}(t,\underline{S}(t)) {B}$. 

\bigskip

The Markovian process $\underline{S}$ verifies: (i) $\underline{S}$ and $\underline{\xi}$ have identical one-dimensional (in time) distributions, i.e. they have identical distributions at every single time $t$ conditional on the common initial condition at time $0$. (ii) The following equality holds:
\begin{equation}
  \mathbb{E}[vv^{T}|\underline{\xi}(t) = \underline{x}]  = \sigma \
\sigma^T(t,\underline{x}).
\end{equation}

\bigskip

To show this, denote $v(t,\underline{\xi}(t)) = diag(\underline{\xi}(t)) {A^I(t)}$ so that
\begin{equation}\label{Markov_projection}
\mathbb{E}[vv^{T}|\underline{\xi}(t) \in d \underline{x}] =
\frac{\mathbb{E}[diag(\underline{\xi}(t))\ \Sigma\
diag(\underline{\xi}(t))\ 1_{\{\underline{\xi}(t) \in
d\underline{x}\}}]}{\mathbb{E}[1_{\{\underline{\xi}(t) \in
 d\underline{x}\}}]}.
\end{equation}
%


\bigskip

Calculate the probability density of $\xi$ as
\begin{eqnarray*}
\mathbb{E}[1_{\{\underline{\xi}(t) \in d \underline{x}\}}] & = & \mathbb{E}\left[ \sum_{k_1,...,k_n =1}^{N} 1_{\{ I_1 = k_1,\ldots,I_n = k_n\}}  1_{\{\underline{\xi}(t) \in d \underline{x}\}}\right]\\
 & = & \sum_{k_1,...,k_n = 1}^{N} \lambda_1^{k_1}...\lambda_n^{k_n}\
\ell_{1,\ldots,n;t}^{k_1,\ldots,k_n} (\underline{x}) \ d\underline{x} 
\end{eqnarray*}
and notice it is the same as the density for MVMD, where we have used independence of $I_i$ of each other and of $W$ to factor the expectation of indicators, and similarly
\begin{eqnarray}
\mathbb{E}[diag(\underline{\xi}(t))\ \Sigma\
diag(\underline{\xi}(t))\ 1_{\{\underline{\xi}(t) \in
d\underline{x}\}}] = & &  \nonumber\\
diag(\underline{x}) \sum_{k_1,...,k_n = 1}^{N}
\lambda_1^{k_1}...\lambda_n^{k_n}\ {V}^{k_1,...,k_n}(t) &
\ell_{1,\ldots,n;t}^{k_1,\ldots,k_n} (\underline{x})\ 
diag(\underline{x}) \ d\underline{x}
 \nonumber\\
 & &
\end{eqnarray}
where ${V}$ had been defined in \eqref{V}.

A substitution in (\ref{Markov_projection}) gives
\begin{equation*}
\mathbb{E}[vv^{T}|\underline{\xi}(t) = \underline{x}]  = \sigma \
\sigma^T(t,\underline{x}).
\end{equation*}

\bigskip

We conclude by invoking Gy\"{o}ngi's Lemma \ref{gyongilemma}.
\end{proof}

\begin{corollary}
The process $\underline{\xi}$ has the same distribution function as
the Markovian process $\underline{S}$ for any time $t.$ Then the
MUVM can be used instead of the MVMD model to price European options if convenient.
\end{corollary}

\begin{remark}
The MUVM features the following interesting properties: explicit dynamics, explicit density function,
semi-analytic formulas for European-style derivatives, and semi-analytic formulas for early exercise derivatives (eg. American Options).
The last property follows via an iterated expectation, with the internal filtration  referencing information at time $\epsilon/2$, and is not shared by the Markovian projection MVMD. See again \cite{general_mixture_diffusion} for a discussion of the univariate case.
\end{remark}

We conclude the analysis of the MUVM model with the following Corollary and Remark.
\begin{corollary}
The MUVM has the same copula function as the MVMD model.
\end{corollary}
\begin{proof}
This is an immediate consequence of the Markovian projection property.
\end{proof}

\begin{remark}{\bf Revisiting the asset- instantaneous covariation decorrelation in MVMD.}
We now further comment on the MVMD result on correlation between assets and their instantaneous variances (squared volatilities) and covariances. As we mentioned in the introduction and as we have seen in detail in Theorem \ref{th:loccovass}, in MVMD we have zero correlation  between assets and instantaneous covariances. This is surprising at first sight, since all instantaneous covariances are deterministic functions of the correlated joint assets. However, the result becomes more intuitive when thinking about the relationship with MUVM. The zero correlation is the best approximation MVMD can attain for its non-Markovian originator MUVM, where instantaneous covariations and assets Brownian shocks are fully statistically independent. 
\end{remark}

\subsection{Markovian projection for the geometric basket dynamics}

Consider now the geometric basket (\ref{weighted_geometric_average}) and set 
$w'_i:= w_i / (w_1 + \ldots + w_n)$, so that we write  
\begin{equation}
B_t = \prod_{i=1}^n S_i^{w'_i}
\label{geomBasket}
\end{equation}
For notation convenience we will omit the index in $w'$, writing simply $w$. $w$ is the row vector of weights in the basket.  
The problem we face now is the following. We may consider the dynamics of $B$ in the MVMD model. Such dynamics if clearly non-Markovian with respect to the filtration generated by $B$ itself. However, we may attempt a Markovian projection by trying to find the local volatility of the basket such that the basket marginal distributions are the same as in the original MVMD model.
The true local volatility for the basket is easily obtained by isolating the diffusion coefficient in 
$d \ln B_t = w\ d \ln(\underline{S}(t))$, where $dS$ follows (\ref{edsMVMDwithindBM1}), and is given by 
$$
\sigma_B(t,\underline{S}) = w\ C(t,\underline{S}) B_\rho \ \   
$$
where we added the index $\rho$ to distinguish the factor matrix $B$  in $B B^T =\rho$ from the basket. Since the basket $B$ is one dimensional, its distribution does not change if we replace the vector $\sigma_B(t,\underline{S}) d \underline{W}$ with 
$\sqrt{ \sigma_B(t,\underline{S})\ \sigma_B(t,\underline{S})^T } dW_1$ where $W_1$ is a scalar standard Brownian motion. This means that we can take as true squared basket volatility the quantity $\sigma_B(t,\underline{S})\ \sigma_B(t,\underline{S})^T =: \sigma_{B}^1(t,\underline{S})^2$
leading to 
\[ \sigma_{B}^1(t,\underline{S})^2 = w \  a(t,\underline{S})  w^T = \sum_{i,j=1}^n a_{i,j}(t,\underline{S})  w_i w_j.   \]
We may now consider the Markovian projection of the true basket dynamics with volatility $\sigma^1_B$ into Markovian one-dimensional diffusions. This is done via Gy{\"o}ngi's lemma above. We assume that the basket drift is not a problem, as it is generally driven by no-arbitrage constraints. We rather focus on the volatility. The local volatility formula from Gy{\"o}ngi's lemma is 
\begin{equation}
\sigma_{B,loc}^2(t,B)=\mathbb{E}\{\sigma^1_B(t,\underline{S}(t))^2 | B(t)=B \} = \frac{\mathbb{E}\{\sigma_B^2(t,\underline{S}(t)) 1_{\{ B_t \in dB \}} \}}{\mathbb{E}\{1_{\{ B_t \in dB \}}\}}.
\label{sigma_Bloc}
\end{equation}
We will now derive a closed form solution expression for this formula in detail, with interesting implications for the final result. 
To compute (\ref{sigma_Bloc}), we will first focus on the denominator, and then on the numerator. 
The calculation of the expectation for the denominator solves all the technical issues for a straightforward expression for the numerator, so that it is indeed best starting from the denominator. Remembering that the multivariate density for $\underline{S}(t)$ in MVMD  is a mixture of multivariate lognormals as in (\ref{mixture_k}), we have
\begin{eqnarray}\label{eq:denomproof}
\mathbb{E}\left\{ 1_{\{ \prod_i S_i(t)^{w_i} \in dB   \}} \right\}=\int{d \underline{y} \, 1_{\{ \prod_i y_i^{w_i} \in dB   \}}}
\sum_{k_1,k_2,\cdots k_n=1}^N
\lambda_1^{k_1} \cdots \lambda_n^{k_n}
\ell_{1,\ldots,n; t}^{k_1,\ldots,k_n}(\underline{y}) 
\end{eqnarray}
Each of these integrals is performed on a multivariate lognormal $\ell_{1,\ldots,n;t}^{k_1,\ldots,k_n}(\underline{y})$. For a generic multi--index $k_1,\ldots,k_n$ in the sum (omitted in the following for ease of notation) the corresponding term can be recast as an integral over a standard $n$-dimensional Gaussian with covariance matrix $\Xi$ defined earlier: denoting by $F_i(t)$  the $t$--forward asset price, and defining $x_i=\ln\frac{S_i}{F_i(t)}+\frac{\Xi_{ii}}{2}$
\begin{equation}
\int d \underline{x} \, 1\left\{ \prod_i F_i^{w_i} \exp \left[-\frac12 \sum_i w_i \Xi_{ii} + \sum_i w_i x_i \right] \in dB \right\} n(\underline{x};\Xi) = \left(- \frac{d}{dB} \int_{D_B} d \underline{x} \, n(\underline{x};\Xi)\right) dB
\label{mainIntegral}
\end{equation}
where 
\begin{equation}
D_B=\{\underline{x} \in R^n \, | \, \prod_i F_i^{w_i} \exp \left[\sum_i w_i x_i \right] -B>0 \} = 
\{\underline{x} \in R^n \, | \, \underline{w} \cdot \underline{x} >\gamma_B \};
\end{equation}
$n(\underline{x};\Xi)$ is the multivariate normal distribution density with zero mean and covariance matrix $\Xi$, calculated at $\underline{x}$, and $\gamma_B$ is defined as 
$$
\gamma_B=\ln \left( \frac{B}{\prod_i F_i^{w_i}\exp \left[-\frac12 \sum_i w_i \Xi_{ii} \right]}\right).
$$
Note that $\gamma_B$ contains all the dependence on the basket value. The term to be differentiated in (\ref{mainIntegral}) is nothing but the integral of a multidimensional Gaussian over a half space (the domain $D_B$) so it is bound to be computed easily.

To calculate it we need a few changes of variable which are purely linear--algebraic. Remember that
$$
n(\underline{x};\Xi)=\frac{1}{(2  \pi)^{\frac{n}{2}} \sqrt{\det \Xi}} \exp [-\frac{1}{2} \underline{x}^T \Xi^{-1} \underline{x}];
$$
diagonalize $\Xi=S^T \Lambda S$, with $\Lambda$ diagonal and $S$ unitary.

Let $\underline{y}=S \underline{x}$. Then, (\ref{mainIntegral}) becomes
$$
\frac{1}{(2  \pi)^{\frac{n}{2}} \sqrt{\det \Xi}} \int_{\widetilde{D}_B} d\underline{y} \, \exp [-\frac{1}{2} \underline{y}^T \Lambda^{-1} \underline{y}]
$$
where now
$$
\widetilde{D}_B=\{\underline{y} \in R^n \, | \, \underline{w}^T S^T \underline{x} >\gamma_B \}.
$$

Now let $\underline{z}=\sqrt{\Lambda^{-1}} \underline{y}$; then the above integral becomes ($|| \, ||$ is the Euclidean norm)
$$
\left\{
\begin{array}{l}
\frac{1}{(2 \pi)^{\frac{n}{2}}} \int_{\bar{D}_B} d\underline{z} \, \exp[-\frac{1}{2} ||\underline{z}||^2], \\
\\
\bar{D}_B=\{\underline{z} \in R^n \, | \, \underline{w}^T S^T \sqrt{\Lambda} \underline{z} >\gamma_B \}.
\end{array}
\right.
$$

Denote by $\Gamma$ any orthonormal matrix such that $\Gamma \underline{w}^T S^T \sqrt{\Lambda} = ||\underline{w}^T S^T \sqrt{\Lambda}|| \hat{e}_n$ , $\hat{e}_n=(0,0,\dots 1)^T$ and define finally $\underline{\xi}=\Gamma \underline{z}$. The integral to be calculated now becomes
\begin{equation}
\left\{
\begin{array}{l}
\frac{1}{(2 \pi)^{\frac{n}{2}}} \int_{\Delta_B} d\underline{\xi} \, \exp[-\frac{1}{2} ||\underline{\xi}||^2], \\
\\
\Delta_B=\{ \underline{\xi} \in R^n \,|\, ||\underline{w}^T S^T \sqrt{\Lambda}|| \,\xi_n >\gamma_B \}.
\end{array}
\right.
\label{finalIntegral}
\end{equation}

(\ref{finalIntegral}) finally becomes
$$
\int_{\frac{\gamma_B}{||\underline{w}^T S^T \sqrt{\Lambda}||}}^{+\infty} d\xi_n \, \frac{1}{\sqrt{2 \pi}} \exp[-\frac{1}{2} \xi_n^2]=1-\Phi\left( \frac{\gamma_B}{||(\underline{w}^T S^T \sqrt{\Lambda}||} \right),
$$
$\Phi$ being a one--dimensional cumulative normal; therefore, by differentiating with respect to $B$, (\ref{mainIntegral}) finally becomes the simple expression
$$
n \left( \frac{\gamma_B}{||\sqrt{\Lambda} S \underline{w} ||} \right) \frac{1}{||\sqrt{\Lambda} S \underline{w}||} \frac{1}{B}
$$
with $n$ denoting the standard one--dimensional Gaussian density.

Note that
$$
||\sqrt{\Lambda} S \underline{w}||^2=\sum_{i,j} w_i w_j \Xi_{ij}
$$
 is nothing but the variance of $B$ in (\ref{geomBasket}).

Restoring the $k$--indexation of the MVMD, with $\underline{k} = (k_1,\ldots,k_n)$, the denominator (\ref{eq:denomproof}) in (\ref{sigma_Bloc}) can be written
$$
\frac{1}{B} \sum_{\underline{k}} \lambda_1^{k_1} \cdots \lambda_n^{k_n} n\left( \frac{\gamma_B^{\underline{k}}}{\sqrt{\sum_{i,j=1}^n w_i w_j \Xi^{(\underline{k})}_{ij}(t)}} \right) \frac{1}{\sqrt{\sum_{i,j=1}^n w_i w_j \Xi^{(\underline{k})}_{ij}(t)}},
$$
which reveals itself as a linear combination of lognormal densities in $B$.
%
%

We can now move to calculating the numerator in (\ref{sigma_Bloc}), namely $\mathbb{E}\{\sigma_B^1(t,\underline{S}(t))^2 1_{\{ B_t \in dB\}} \}=$
\begin{eqnarray*}
=\int d \underline{y} \, 1\left\{ \prod_i y_i^{w_i} \in dB\right\} \sum_{i,j=1}^n w_i w_j a_{ij}(t,\underline{y}) \sum_{\underline{k}}
\lambda_1^{k_1} \cdots \lambda_n^{k_n}
\ell_{1,\ldots,n; t}^{k_1,\ldots,k_n}(\underline{y}) \\
=\sum_{i,j=1}^n w_i w_j \sum_{\underline{k}} \lambda_1^{k_1} \cdots \lambda_n^{k_n} V_{ij}^{(\underline{k})}(t) \int{d \underline{y} \,1 \left\{\prod_i y_i^{w_i} \in dB\right\} \ell_{1,\ldots,n; t}^{k_1,\ldots,k_n}(\underline{y})}
\end{eqnarray*}
where we have used (\ref{C}) and we have the same type of integrals as before. 
The numerator then becomes  
\begin{eqnarray*}
\mathbb{E}\{\sigma^1_B(t,\underline{S}(t))^2 1\{ B_t \in dB\} \}=\sum_{i,j=1}^n w_i w_j \sum_{\underline{k}}  \lambda_1^{k_1} \cdots \lambda_n^{k_n} V_{ij}^{(\underline{k})}(t) n\left( \frac{\gamma_B^{(\underline{k})}}{||\sqrt{\Lambda^{(\underline{k})}} S^{(\underline{k})} \underline{w}||} \right) \frac{1}{||\sqrt{\Lambda^{(\underline{k})}} S^{(\underline{k})} \underline{w}||} \frac{dB}{B} \\
=\frac{dB}{B} \sum_{\underline{k}} \lambda_1^{k_1} \cdots \lambda_n^{k_n} n\left( \frac{\gamma_B^{(\underline{k})}}{||\sqrt{\Lambda}^{(\underline{k})} S^{(\underline{k})} \underline{w}||} \right) \frac{1}{||\sqrt{\Lambda^{(\underline{k})}} S \underline{w}||} \sum_{i,j=1}^n w_i w_j V_{ij}^{(\underline{k})}(t) \\
\end{eqnarray*}
We have thus proven the following
\begin{theorem}{\bf Markovian projection of the MVMD basket dynamics.}
The squared local volatility for geometric basket dynamics associated with the MVMD model is given by
\begin{equation}
\sigma_{B,loc}^2(B,t)=\frac{\sum_{\underline{k}} \lambda_1^{k_1} \cdots \lambda_n^{k_n}     n\left( \frac{\gamma_B^{(\underline{k})}}{\sqrt{\sum_{i,j=1}^n w_i w_j \Xi_{ij}^{(\underline{k})}(t)}} \right) \frac{1}{\sqrt{\sum_{i,j=1}^n w_i w_j \Xi_{ij}^{(\underline{k})}(t)}} \sum_{i,j=1}^n w_i w_j V_{ij}^{(\underline{k})}(t)}{\sum_{\underline{k}} \lambda_1^{k_1} \cdots \lambda_n^{k_n}  n\left( \frac{\gamma_B^{(\underline{k})}}{\sqrt{\sum_{i,j=1}^n w_i w_j \Xi_{ij}^{(\underline{k})}(t)}} \right) \frac{1}{\sqrt{\sum_{i,j=1}^n w_i w_j \Xi_{ij}^{(\underline{k})}(t)}}},
\label{geomBasketLocalVariance}
\end{equation}
which is the analogous for $B$ of the original univariate LMD model volatility $s$ in (\ref{eq:dcmix}). 
In particular, remembering the expression for $\gamma_B^{(k)}$, we recognize a locally weighted average of the basket instantaneous variance calculated over the many Black--Scholes states that the mixture is based upon.
There is therefore a mixture consistency result at work:  whenever we consider a geometric basket on a MVMD model, the Markovian projection of this basket dynamics onto a univariate diffusion yields precisely the one-dimensional LMD model that served as inspiration for MVMD. The same result does not hold for arithmetic baskets. 
\end{theorem}
In this theorem we turned the usual reasoning on its head: we know that under the assumption that the densities of the geometric basket $B$ are mixtures of lognormals with constant coefficients $\lambda_k$, $B$'s local variance will indeed take the form \eqref{geomBasketLocalVariance}; but under the usual assumptions there exists a unique strong solution for the corresponding SDE as we have seen in Theorem~\ref{th:LMDE}.

\bigskip

In the case of an equity index, or for that matter any other index based on constant weights, the alternative possibilities to construct the local volatility are then
\begin{itemize}
\item to approximate the index with the corresponding geometric basket throughout the calculation, or 
\item to use \eqref{geomBasketLocalVariance} for the local volatility of the geometric basket as a proxy for the local volatility of the index, of course correcting $B$ for the mismatch between the arithmetic and the geometric average at time 0, \`a la Kemna--Vorst \cite{Vorst}.
\end{itemize}

\section{Option pricing}\label{sec:optpric}

Suppose that
$\underline{S}$ represents the vector of underlying asset prices composing the
underlying $B$ in Eq. (\ref{basket}) or Eq.
(\ref{weighted_geometric_average}). A conventional scheme for
pricing a plain option on the underlying basket in a way fully consistent
with individual local volatilities would require, according to a SCMD type approach, a
sufficiently fine time discretization coupled with a Monte Carlo
integration with instantaneous covariance given by Eq.
(\ref{factorVol1}) (or by more complicated discretization schemes for
SDEs, see e.g. Milstein's \cite{kloeden_platen}). Our MVMD model
allows instead to compute the option price (\ref{European_option_price}) 
through a set of single--step Monte Carlo integrations (one
integration for each combination $(k_1, \cdots ,k_n)$). Indeed since
the terminal distribution of ${\underline{S}}(T)$ is known, the MVMD model allows to evaluate
simple claims on a basket without time discretization. Thus, using
the MVMD approach, one can reduce the computational time
significantly. But the actual consequences of this approach are
wider, in that they affect the many--body dynamics in a deeper way.

\bigskip

Remembering (\ref{mixture_k}), it is straightforward to obtain the
model option prices in terms of the option prices associated to the
instrumental processes (momentarily thought of as underlying assets) $(Y_i^k)_{i = 1,\cdots,n, k = 1,\cdots, N}.$

\subsection{Option on an arithmetic basket}

Let us begin by considering an option of European type on the basket of securities of Eq. (\ref{basket})
with maturity $T$ and strike $K.$ The risk free interest rate is denoted by $r$ and is assumed to be constant for simplicity.
Then, if $\omega = 1$ for a call
and $\omega = -1$ for a put, the option value
(\ref{European_option_price}) can be written as
\begin{equation}\label{basket_option_equ}
\Pi = e^{-rT}  \int_{\mathbb{R}^n} \left[\omega(\sum_{k = 1}^n w_k
x_k - K)\right]^{+} p_{\underline{S}(T)} (x_1,\ldots,x_n) dx_1\ldots
dx_n
\end{equation}
where $p_{\underline{S}(T)}$ is the joint density of the random
variables $S_1(T)$,$\ldots$, $S_n(T)$ and is given by Eqs.
(\ref{mixture_k}) --(\ref{xtild_dimn}). Finally we have
\begin{equation}\label{basket_option_price_mixture}
\Pi = \sum_{k_1,\ldots, k_n = 1}^N
\lambda_1^{k_1}\ldots\lambda_n^{k_n}\ \ \Theta_{k_1,\ldots, k_n}
\end{equation}
where $\Theta_{k_1,\ldots, k_n}$ denotes the European option price
associated to the basket $\sum_{i = 1}^n w_i Y_i^{k_i}.$

\bigskip

When the value of the basket (\ref{basket}) contains short positions as well then
we are dealing with spread options.

\subsection{Spread option}

The simplest spread option is an option of the European type on the
difference of two underlying assets. The spread is naturally defined
as the instrument 
%
\begin{equation}\label{spread}
B(t) = S_2(t) - S_1(t),\  \          t\geq 0.
\end{equation}
Buying such a spread is buying $S_2$ and selling $S_1$.

\bigskip

The price of the simplest spread option is a particular case of
(\ref{basket_option_price_mixture}) and equal to
\begin{equation}\label{spread_mixture_price}
\Pi = \sum_{i,j = 1}^N \lambda_1^i \lambda_2^j\\ \ \Theta_{i,j}
\end{equation}
where $\Theta_{i,j}$ denotes the European spread option price
associated to the instrumental spread $Y_2^j - Y_1^i.$

\bigskip

For all $i,j = 1,\ldots,N$, $Y_1^i$ and $Y_2^j$ are log-normal
underlying assets evolving according to the SDE (\ref{edsSi_k}). Let
us denote the correlation coefficient between the two assets by $\rho$.
It is possible to give a Black--Scholes type formula for the price
of the European option with maturity $T$ associated to the spread
$Y_2^j - Y_1^i$ when the strike is $K = 0$, provided that the drifts
$\mu_1 = \mu_2 = r$ match the short interest rate $r$ and the volatilities
$\sigma_1^i$ and $\sigma_2^j$ are constant in time. This is of course Margrabe's 1978 formula
\cite {margrabe}. It cannot be extended to the general case $K
\neq 0$ (but the price in that case can easily be computed by a one--dimensional numerical integration.)
Besides the fact that the case $K = 0$ leads to a solution in fully closed
form, it has also a practical appeal to the market participants.
Indeed, it can be viewed as an option to exchange one asset for
another at no additional cost.
\begin{proposition}
When the strike $K = 0,$ the European spread option price is also the price of an option to exchange one asset $S_1$ for another $S_2$, and under the MVMD model is given by Formula \eqref{spread_mixture_price}, where
$\Theta_{i,j}$ is given by
\begin{equation} \label{margrabe_formula}
\Theta_{i,j} = \omega \left[x_2^j \Phi(\omega d_1^{ij}) - x_1^i
\Phi(\omega d_0^{ij}) \right],
\end{equation}
where $$d_1^{ij} = \frac{\ln(x_2^j/x_1^i)}{\sigma^{ij} \sqrt{T}} +
\frac{1}{2} \sigma^{
ij} \sqrt{T},\ \  d_0^{ij} =
\frac{\ln(x_2^j/x_1^i)}{\sigma^{ij} \sqrt{T}} - \frac{1}{2}
\sigma^{ij} \sqrt{T}$$ and $x_1^i = Y_1^i(0)$, $x_2^j = Y_2^j(0)$,
$(\sigma^{ij})^2 = (\sigma_1^{i })^2 - 2 \rho \sigma_1^i \sigma_2^j +
(\sigma_2^{j })^2, \Phi$ is the standard normal cumulative distribution
function, $T$ the maturity, and $\omega = 1$ for a call and $\omega =
-1$ for a put.
\end{proposition}
The proof is straightforward.
\subsection{Option on a geometric basket}
Let us consider an option of European type on the basket of
securities of Eq. (\ref{weighted_geometric_average}) with maturity
$T$ and strike $K.$ The short-term interest rate is denoted by $r$ and is assumed to be a deterministic constant. Then, if $\omega =1$ for a call and $\omega = -1$ for
a put, the option value (\ref{European_option_price}) can be written
as
\begin{equation}
\Pi = e^{-r T} \int_{\mathbb{R}^n}\left\{
\omega\left[\left(x_1^{w_1}\cdots x_n^{w_n}\right)^{\frac{1}{w_1 +
\cdots + w_n}} - K\right]\right\}^+
p_{\underline{S}(T)}(\underline{x}) dx_1\cdots dx_n
\end{equation}
where $p_{\underline{S}(T)}$ is the joint density of the random
variables $S_1(T),\cdots, S_n(T)$ and is given by Eqs.
(\ref{mixture_k}) --(\ref{xtild_dimn}). We have that
\begin{equation} \label{geometric_basket_option_price_mixture}
\Pi = \sum_{k_1,\cdots,k_n = 1}^N \lambda_1^{k_1} \cdots
\lambda_n^{k_n}\  \Gamma_{k_1,\cdots,k_n}
\end{equation}
where $\Gamma_{k_1,\cdots,k_n}$ denotes the European option price at
initial time $t = 0$ associated to the instrumental geometric-average basket  $\left(Y_1^{k_1^{w_1}}
\cdots Y_n^{k_n^{w_n}}\right)^{\frac{1}{w_1 + \cdots + w_n}}.$ Since this geometric average is based on
lognormal instrumental variables it is itself lognormal, and leads to Black Scholes type closed form formulas for the $\Gamma$ terms.
%
%
Let us now consider the particular case $n = 2.$ The European option
on weighted geometric average is then equal to
\begin{equation}\label{asian_option_price_n2}
\Pi = \sum_{i,j = 1}^N \lambda_1^{i} \lambda_2^{j}\ \Gamma_{i,j}
\end{equation}
where $\Gamma_{i,j}$ denotes the European option price at initial
time $t = 0$ associated to the instrumental geometric basket $\left(Y_1^{i^{w_1}}
Y_2^{j^{w_2}} \right)^{\frac{1}{w_1 + w_2}}.$
Recall that $Y_1^i$ and $Y_2^j,$ $\forall i,j = 1,\cdots, N$ are
lognormal underlying assets evolving according to the SDE
(\ref{edsSi_k}). If the drifts $\mu_1 = \mu_2 = r$ and the
volatilities $\sigma_1^i$ and $\sigma_2^j$ are constants in time,
the price of the European Call option with maturity $T$ associated
to the basket $\left(Y_1^{i^{w_1}}
Y_2^{j^{w_2}}\right)^{\frac{1}{w_1 + w_2}}$ when $K = 0$ is given by
a closed form formula.
\begin{proposition}
In the case $n=2$ and with strike $K = 0$, the price of a European Call option on a 
geometric  basket under the MVMD model is given by Formula \eqref{asian_option_price_n2} where   $\Gamma_{i,j}$ is given by
\begin{equation}
\begin{array}{l} \Gamma_{i,j} = \exp(-r T) Y_1^i(0) ^{\varpi
w_1} Y_2^j(0)^{\varpi w_2} \exp \left\{\left[(r - \frac{1}{2}
\sigma_1^{i2}) w_1 + (r - \frac{1}{2} \sigma_2^{j2}) w_2
\right]\varpi T + \right.
\\
 \\
\hspace{35mm}\left.\frac{1}{2}\left[\sigma_1^{i2} w_1^2 +
\sigma_2^{j2} w_2^2 + 2 \rho \sigma_1^i \sigma_2^j w_1
w_2\right]\varpi^2 T\right\}
\end{array}
\end{equation}
where $\varpi = \frac{1}{w_1 + w_2}$ and $\rho$ denotes the
correlation coefficient between $Y_1^i$ and $Y_2^j.$
\end{proposition}
\begin{proof}
To ease the notation we shall omit indices $i,j$.

\begin{equation}
\begin{array}{ll}
\Gamma & =  e^{-r T}\ \mathbb{E}\left\{\left[Y_1(T)^{w_1}
Y_2(T)^{w_2}\right]^{\frac{1}{w_1 +
w_2}}\right\}\\
 & \\
& =  e^{-r T} Y_1(0)^{\varpi w_1} Y_2(0)^{\varpi w_2} e^{\left[(r -
\frac{1}{2} \sigma_1^{2}) w_1 + (r - \frac{1}{2} \sigma_2^{2}) w_2
\right]\varpi T}\ \mathbb{E}\left[e^{\gamma Z}\right]
\end{array}
\end{equation}
where $Z$ is a standard normal variable and
$$\gamma =
\sqrt{\left[\sigma_1^{2} w_1^2 + \sigma_2^{2} w_2^2 + 2 \rho
\sigma_1 \sigma_2 w_1 w_2\right]\varpi^2 T}.$$ The result follows by
using $\mathbb{E} (e^{\gamma Z}) = e^{\gamma^2 / 2}.$

\end{proof}
\begin{remark}
The derivations (\ref{basket_option_price_mixture}) and
(\ref{geometric_basket_option_price_mixture}) show that a dynamics
leading to an $n$-dimensional density for the vector of underlying asset prices
that is the convex combination of $n$-dimensional basic densities
induces the same convex combination among the corresponding option
prices. Furthermore, due to the linearity of the derivative
operator, the same convex combination applies to option Greeks such as delta or gamma.
\end{remark}
\begin{remark}
The results of this section can be easily extended to hold in the case of shifted lognormal densities \cite{mixtureFX}.
\end{remark}

\section {Numerical Results: SCMD vs MVMD}\label{sect:numerical_results_pricing}

In this section we present some results for the pricing of three
typical options: European Call on a weighted arithmetic average
containing only long positions, European Call Spread option (long and short positions) and
European Call option on  a weighted geometric average of a basket.
We investigate these options in the 
SCMD and MVMD frameworks in order to compare them. The performance
of our approach is investigated by comparing the prices under the
two models.

\bigskip

For numerical sake, we focus on the two dimensional case $n=2$ where each
individual component of the asset is modeled with a mixture of two
lognormal densities, $N=2$. We assume also that the short-term interest rate $r$ is deterministic and constant throughout the life of the option (i.e., until the maturity
date $T$). Then, from Eq. (\ref{European_option_price}), the
European Call prices tested in this section are given by the
risk--neutral expectation
\begin{equation}\label{europprice_rcst}
\Pi = e^{-rT} \mathbb{E}\left[\left(B_T - K\right)^+ \right]
\end{equation}
where $B$ is the underlying basket instrument at maturity $T.$

\subsection{Arithmetic basket and spread options} \label{subsect_arithmetic_average}

The European Call prices tested in this section are given by
(\ref{europprice_rcst}) where $B$ is given by (\ref{basket}) with
$(w_k)_{k=1,2}
> 0$ for the option on  a weighted arithmetic average containing only
long positions. We call this option "Vanilla basket". Instead, $B$ is
given by Eq. (\ref{spread}) for the spread option.

Note that, under MVMD, the vanilla basket option price is given by Eq.
(\ref{basket_option_price_mixture}) with $n = N = 2$ and the spread
option price is given by Eq. (\ref{spread_mixture_price}) with $N =
2$.

\bigskip

The parameters of the test baskets are given in Table
\ref{basket_parameters}. The interest rate $r$ is $5\%.$ The time to
maturity $(T)$ is one year. The strike $K$ takes the values $K =
0.7$, $K = 1$ and $K = 1.3$. In order to obtain the fair price of
the options under SCMD, 100,000 Monte Carlo runs are performed and an Euler
scheme with time step $\Delta t = 1/360$ is applied. The first
comparison uses a correlation $\rho = 0.6$. The results are given in
Table \ref{call_on_basket_rho06}. The second comparison is done for
a correlation $\rho = 1$. The results are shown in Table
\ref{call_on_basket_rho1}. The standard error of the prices is given
in parentheses.


\begin{table}[htb]
\begin{center}
\begin{tabular}{|c|c|c|}
\hline
 & Vanilla Basket & Spread\\
\hline Initial prices ([$S_1(0),S_2(0)$])& [1,1] & [0.7,1.7]\\
\hline drift ([$\mu_1$,$\mu_2$]) & [5 \%,5 \%] & [5 \%,5 \%]\\
\hline [$\lambda_1^1$,$\lambda_1^2$] & [0.6,0.4] & [0.6,0.4]\\
\hline
[$\lambda_2^1$,$\lambda_2^2$] & [0.7,0.3] & [0.7,0.3]\\
\hline
[$\sigma_1^1$,$\sigma_1^2$] & [0.3,0.2] & [0.2,0.1]\\
\hline [$\sigma_2^1$,$\sigma_2^2$] & [0.25,0.35] & [0.4,0.5]\\
\hline
weights $[w_1,w_2]$ & [0.5,0.5] & [-1,1]\\
\hline
\end{tabular}
\end{center}
\caption{Basket option parameters} \label{basket_parameters}
\end{table}
\begin{table}[htb]
\begin{center}
\begin{tabular}{|c|c|c|}
\hline
\multicolumn {3}{|c|}{$K = 0.7$}\\
\hline
 & Vanilla Basket & Spread\\
\hline
MVMD & 0.3380 (0.0007) & 0.4413 (0.0019)\\
SCMD & 0.3386 (0.0007) & 0.4365 (0.0019)\\
\hline \hline
\multicolumn {3}{|c|}{$K = 1$}\\
\hline
MVMD & 0.1202 (0.0005) & 0.2868 (0.0017)\\
SCMD & 0.1200 (0.0005) & 0.2833 (0.0017)\\
\hline \hline
\multicolumn {3}{|c|}{$K = 1.3$}\\
\hline
MVMD & 0.0290 (0.0003) & 0.1810 (0.0014)\\
SCMD & 0.0296 (0.0003) & 0.1836 (0.0014)\\
\hline
\end{tabular}
\end{center}
\caption{European Call on Basket Prices and Standard Errors for $\rho
= 0.6$} \label{call_on_basket_rho06}
\end{table}
\begin{table}[htb]
\begin{center}
\begin{tabular}{|c|c|c|}
\hline \multicolumn {3}{|c|}{$K = 0.7$}\\
\hline
 & Vanilla Basket & Spread\\
\hline
MVMD & 0.3404 (0.0008) & 0.4199 (0.0018)\\
SCMD & 0.3411 (0.0008) & 0.4193 (0.0019)\\
\hline \hline
\multicolumn {3}{|c|}{$K = 1$}\\
\hline
MVMD & 0.1307 (0.0006) & 0.2611 (0.0016)\\
SCMD & 0.1305 (0.0006) & 0.2647 (0.0016)\\
\hline \hline
\multicolumn {3}{|c|}{$K = 1.3$}\\
\hline
MVMD & 0.0364 (0.0003) & 0.1661 (0.0013)\\
SCMD & 0.0373 (0.0003) & 0.1637 (0.0013)\\
\hline
\end{tabular}
\end{center}
\caption{European Call on Basket Prices and Standard Errors for $\rho
= 1$} \label{call_on_basket_rho1}
\end{table}

\newpage

In Proposition \ref{th:taumd} we derived a closed form formula (\ref{taumix}) for Kendall's tau in a normal mixture. This formula can be easily generalized to compute Kendall tau for the MVMD model. Through this formula (or alternatively simulation) for MVMD and simulation for SCMD, we now compare Kendall's tau for MVMD and SCMD when the parameters are assumed to be the same.

The initial parameters we use are given in Table \ref{initial_parameters}.

\bigskip

\begin{table}[htb]
\begin{center}
\begin{tabular}{|c|c|}
\hline
$S_1(0)$ = 1 & $S_2(0)$ = 1\\
\hline $\mu_1$ = 5 \% & $\mu_2$ = 3 \% \\
\hline
$\sigma_1^1$ = 0.3 & $\sigma_2^1$ = 0.25\\
\hline
$\sigma_1^2$ = 0.2 & $\sigma_2^2$ = 0.35\\
\hline
$\lambda_1^1$ = 0.6 & $\lambda_2^1$ = 0.7\\
\hline
$\lambda_1^2$ = 0.4 & $\lambda_2^2$ = 0.3\\
\hline
\end{tabular}
\end{center}
\caption{Initial parameters} \label{initial_parameters}
\end{table}

Computing Kendall's tau under SCMD requires the choice of a
discretization time step $\Delta t$, and the generation of discrete
time samples $\underline{S}(t_0 + j \Delta t)$ for $j = 0, 1, . . .
,M$ with $t_0 = 0$ and $t_0 + M \Delta t = T.$ The discretization
time steps $\Delta t$ should be taken with great care to make sure
that the numerical scheme used to generate the discrete samples
produce reasonable approximations. A good choice is an Euler scheme
over equal time steps of size $\Delta t = 1/360.$

The first comparison uses $\rho =
0.6.$ The results are given in Table
\ref{kendall_volatilitycons0.6}. The next comparison is done for
$\rho = -0.6$ . The results are given in Table \ref{kendall_volatilitycons_0.6}. The final comparison uses $\rho = 1.$ The results are shown in Table \ref{kendall_volatilitycons1}. The standard error value is given in parentheses.
\bigskip

\begin{table}[htb]
\begin{center}
\begin{tabular}{|c|c|c|c|}
\hline
 Maturity & Exact $\tau$ for MVMD  & $\tau$ simulation under MVMD & $\tau$ simulation under SCMD\\
\hline 1 y & 0.4016 & 0.4012 (0.0004) & 0.4092 (0.0004)\\
\hline 5 y & 0.3977& 0.3976 (0.0004) & 0.4093 (0.0004)\\
\hline 10 y & 0.3929 & 0.3930 (0.0004) & 0.4090 (0.0004)\\
\hline
\end{tabular}
\end{center}
\caption{Kendall's tau ($\tau$) under MVMD vs SCMD and Standard Errors (in parentheses) for $\rho = 0.6.$} \label{kendall_volatilitycons0.6}
\end{table}
\begin{table}[htb]
\begin{center}
\begin{tabular}{|c|c|c|c|}
\hline
 Maturity & Exact $\tau$ for MVMD  & $\tau$ simulation under MVMD & $\tau$ simulation under SCMD\\
\hline 1 y & -0.4016 & -0.4018 (0.0004) & -0.4084 (0.0004)\\
\hline 5 y & -0.3976 & -0.3976 (0.0004) & -0.4091 (0.0004)\\
\hline 10 y & -0.3927 & -0.3928 (0.0004) & -0.4090 (0.0004)\\
\hline
\end{tabular}
\end{center}
\caption{Kendall's tau ($\tau$) under MVMD vs SCMD and Standard Errors (in parentheses) for $\rho =  -0.6.$} \label{kendall_volatilitycons_0.6}
\end{table}
\begin{table}[htb]
\begin{center}
\begin{tabular}{|c|c|c|c|}
\hline
 Maturity & Exact $\tau$ for MVMD  & $ \tau$ simulation under MVMD & $\tau$ simulation under SCMD\\
\hline 1 y & 0.9109 & 0.9112 (0.0002) & 0.9940 (0.00004)\\
\hline 5 y & 0.8893 & 0.8894 (0.0002) & 0.9949 (0.00004)\\
\hline 10 y & 0.8650 & 0.8648 (0.0002) & 0.9950 (0.00004)\\
\hline
\end{tabular}
\end{center}
\caption{Kendall's tau ($\tau$) under MVMD vs SCMD and Standard Errors (in parentheses) for $\rho =  1.$} \label{kendall_volatilitycons1}
\end{table}

We see that there is more terminal dependence in absolute value in SCMD than in MVMD. In the SCMD Kendall's tau does not change with the maturity, whereas, its absolute value goes down significantly as the maturity increases in the MVMD model. The relative difference of Kendall's tau between SCMD and MVMD is
increasing with the maturity. It is relatively limited when $\rho =
\pm 0.6$ and we could see more of difference when $\rho = 1$. We will analyze this more in depth in further work, but this result is reminiscent of our correlation analysis in our earlier Corollary \ref{co:loccorr}.
\bigskip

Despite this difference,
the basket option price is not very sensitive to the difference between the two models, and indeed Table
\ref{call_on_basket_rho1} shows that the prices obtained by the two
models are close. Table \ref{call_on_basket_rho06} shows that this
feature is maintained for a correlation $ \rho = 0.6$. Notice that
the prices obtained by the two models when dealing with a basket
option with long positions are closer than in the case of a spread option. The price of the basket option with long
positions increases with the correlation between the assets
for all strikes whereas the price of the spread option decreases. Obviously, increasing the strike decreases dramatically the
prices of both options in the two models for all values
of correlation. The price of the spread option is higher than the
price of the basket option with long positions and the difference
between the two option prices becomes smaller as the correlation
increases. These features hold
for all strikes in the MVMD and SCMD models and are quite reasonable.

\bigskip

These results seem to suggest that an option on an arithmetic basket containing only long
positions and a spread option are not affected in an extreme way by the
dependence between the different assets since even models that give
different Kendall's tau give quite similar prices.

The largest relative difference we find in our pricing examples is for the spread option when $\rho=1$ and $K=1.3$, see Table \ref{call_on_basket_rho1} (last two rows, last column). In this case the relative difference between the MVMD price and the SCMD price is about $1.4\%$. However the difference for the corresponding Kendall tau's in MVMD and SCMD, as given in Table \ref{kendall_volatilitycons1} (first row, last two columns), is about $9\%$. Hence we see that to a large relative difference in the dependence structure corresponds a much smaller relative difference in option prices. 

Finally, since the MVMD and SCMD models give similar numerical results in pricing
European Call option on a weighted arithmetic average containing only long
positions and European Call Spread option, the MVMD model is the
most convenient here since it allows to compute the option price in one single Monte-Carlo step which can then be evaluated rapidly.

\bigskip

In the next section, we will price an European Call option on a weighted
{\emph{geometric}} average under the SCMD and MVMD models and
investigate if this option is more sensitive to the different statistical dependence between the two models.

\subsection{Geometric basket option}

The European Call price tested in this paragraph is given by
(\ref{europprice_rcst}) where $B$ is given by Eq.
(\ref{weighted_geometric_average}). Note that, under MVMD, this
option price is given by Eq. (\ref{asian_option_price_n2}) with $N =
2$.

\bigskip

Table \ref{geom_basket_parameters} reports the parameters we use.
The interest rate $r$ is $5\%.$ The time to maturity $(T)$ is one
year. The strike $K$ takes the values $K = 0.7,$ $K = 1$ and $K =
1.3$. In order to obtain the fair price of the options 100,000 Monte
Carlo runs are performed and an Euler scheme with time step $\Delta
t = 1/360$ is applied. The first comparison uses a correlation $\rho
= 0.6$. The results are given in Table \ref{geometric_basket_rho06}.
The second comparison is done for a correlation $\rho = -0.6$. The
results are reported in Table \ref{geometric_basket_rho_minus06}. A
last comparaison uses a correlation $\rho = 1.$ The results are
shown in Table \ref{geometric_basket_rho1}.

\bigskip

\begin{table}[htb]
\begin{center}
\begin{tabular}{|c|c|}
\hline Initial prices ([$S_1(0),S_2(0)$])& [1,1]\\
\hline drift ([$\mu_1$,$\mu_2$]) & [5 \%,5 \%]\\
\hline [$\lambda_1^1$,$\lambda_1^2$] & [0.6,0.4]\\
\hline
[$\lambda_2^1$,$\lambda_2^2$] & [0.7,0.3]\\
\hline
[$\sigma_1^1$,$\sigma_1^2$] & [0.3,0.2]\\
\hline [$\sigma_2^1$,$\sigma_2^2$] & [0.25,0.35]\\
\hline
weights $[w_1,w_2]$ & [1,1]\\
\hline
\end{tabular}
\end{center}
\caption{Basket Option parameters} \label{geom_basket_parameters}
\end{table}

\begin{table}[htb]
\begin{center}
\begin{tabular}{|c|c|}
\hline
\multicolumn {2}{|c|}{$K = 0.7$}\\
\hline
 &  Option price\\
\hline
MVMD & 0.3313 (0.00074) \\
SCMD & 0.3312 (0.00075) \\
\hline \hline
\multicolumn {2}{|c|}{$K = 1$}\\
\hline
MVMD & 0.1154 (0.00055) \\
SCMD & 0.1159 (0.00057) \\
\hline \hline
\multicolumn {2}{|c|}{$K = 1.3$}\\
\hline
MVMD & 0.0267 (0.00028) \\
SCMD & 0.0268 (0.00029) \\
\hline
\end{tabular}
\end{center}
\caption{European Call on Basket Prices and
Standard Errors (in parentheses) for $\rho = 0.6$} \label{geometric_basket_rho06}
\end{table}
\begin{table}[htb]
\begin{center}
\begin{tabular}{|c|c|}
\hline
\multicolumn {2}{|c|}{$K = 0.7$}\\
\hline
 &  Option price\\
\hline
MVMD & 0.3049 (0.00037) \\
SCMD & 0.3045 (0.00037) \\
\hline \hline
\multicolumn {2}{|c|}{$K = 1$}\\
\hline
MVMD & 0.0584 (0.00025) \\
SCMD & 0.0574 (0.00025) \\
\hline \hline
\multicolumn {2}{|c|}{$K = 1.3$}\\
\hline
MVMD & 0.0016 (0.00003) \\
SCMD & 0.0013 (0.00003) \\
\hline
\end{tabular}
\end{center}
\caption{European Call on Basket Prices and
Standard Errors (in parentheses) for $\rho = - 0.6$}
\label{geometric_basket_rho_minus06}
\end{table}
\begin{table}[htb]
\begin{center}
\begin{tabular}{|c|c|}
\hline
\multicolumn {2}{|c|}{$K = 0.7$}\\
\hline
 &  Option price\\
\hline
MVMD & 0.3387 (0.00083) \\
SCMD & 0.3413 (0.00084) \\
\hline \hline
\multicolumn {2}{|c|}{$K = 1$}\\
\hline
MVMD & 0.1308 (0.00063) \\
SCMD & 0.1307 (0.00064) \\
\hline \hline
\multicolumn {2}{|c|}{$K = 1.3$}\\
\hline
MVMD & 0.0367 (0.00035) \\
SCMD & 0.0376 (0.00038) \\
\hline
\end{tabular}
\end{center}
\caption{European Call on Basket Prices and
Standard Errors (in parentheses) for $\rho = 1$} \label{geometric_basket_rho1}
\end{table}

While Kendall's tau is different between the SCMD and MVMD models especially when $\rho$ is high, as we have seen eaerlier, the option price is not as sensitive. Table \ref{geometric_basket_rho1} shows that the prices obtained by the two models are close. Tables \ref{geometric_basket_rho06} and \ref{geometric_basket_rho_minus06} show that this feature is maintained for a correlation $\rho = \pm 0.6$. We see that the prices obtained by the two models are close to (but less than, see \cite{Vorst}) those obtained previously in Section  \ref{subsect_arithmetic_average} when dealing with an option on a weighted arithmetic average of a basket with long positions. All the experiments show that the price of the option increases as the correlation between the assets increases for all strikes. It can be seen that increasing the strike decreases dramatically the prices in the two models for the different values of correlation.

\bigskip

These results seem to suggest that an option on a weighted geometric average of a basket is not very sensitive to dependence between the different assets since even models that give different Kendall's tau give quite similar prices. This is basically the same result we obtained for the arithmetic average basket in the previous section.

Because the above observations show that the MVMD and SCMD models give similar numerical results in European Call option on a weighted geometric average of a basket pricing, it is better to use the MVMD model allowing to compute the option price in closed form.

\section{Conclusions and perspectives}

We illustrated how to extend in a conceptually simple fashion an
asset price model, the so--called (univariate and possibly shifted) lognormal mixture dynamics, that has been shown to reproduce well general implied
volatility structures commonly observed on the market
\cite{mixture1,mixture2,mixtureFX,sartorelli,fengler,musiela}. This model is formulated in the space of the so--called local volatility models. The extension aims at inferring an analytic expression for the local volatility of a multivariate
security (such as e.g. a basket of underlying assets) that is consistent with
{\it (i)} the individual dynamics of each component of the security
as deduced by that security volatility smile and {\it (ii)} a given instantaneous correlation structure between different securities.

A na\"{i}ve approach would consist in connecting univariate lognormal mixure dynamics models for each asset through an instantaneous correlation connecting the Brownian motions driving different asset dynamics. We refer to this approach as simply correlated mixture dynamics, SCMD.

However, we improve this approach by extending the mixture dynamics to the multivariate case in a more radical way, leading to the multi-variate mixture dynamics, MVMD, implying a multivariate mixture rather than single univariate mixtures patched together by Brownian correlations. While this is perfectly equivalent to SCMD for single assets, the main practical advantage of our MVMD extension is that our approach allows for a semi-analytic pricing of European style derivatives on the multivariate security in a way that takes into account the smile structures of the individual component securities and reduces
computational time, while staying arbitrage--free. Another important advantage is the availability of closed--form dependence measures, that are important in a multi-asset setting.  This points to MVMD  being an arbitrage-free dynamical model with a great potential for consistently modelling single assets' and baskets' (or indices') volatility smiles.

We further introduced Markovian projection results showing how our model is related to multivariate uncertain volatility models and also illustrating how the Markovian projection for a geometric basket dynamics is consistent with a univariate mixture dynamics model. 

In the paper we also showed that our approach performs remarkably well in terms of basket option pricing with a smile structure of implied volatilities, and provided a number of numerical examples. 
\bigskip

Future extensions
include the testing of this approach in actual situations as swap
rates derivatives within the LIBOR Market Model. Such an extension
would allow computing in a quasi--analytical fashion the swap rates
smile given the smiles in the individual caplets and an instantaneous correlation assumption. We may also apply this setup to triangular relationships among exchange rates in the FX market. An interesting application would be to apply the framework in this paper to a real equity index smile, trying to connect said smile with the index component single smiles. More generally, we could study other payouts whose valuation depends crucially on dependence assumptions, such as best-of baskets and similar products.


\end{document}